\documentclass[authoryear,english,11pt]{elsarticle} 

\usepackage{epsfig}
\usepackage{subfigure}
\usepackage{caption}
\usepackage{subcaption}
\usepackage{amssymb,amsmath,amsfonts,layout,graphicx,amsthm}
\usepackage{makeidx}
\usepackage{babel}
\usepackage{tikz}
\usepackage{sublabel}
\usepackage{cases}
\usepackage{algorithm}
\usepackage{algorithmic}
\usepackage{color}
\usepackage{comment}
\usepackage{booktabs}
\usepackage{longtable}
\usepackage[title]{appendix}
\usepackage{bm}
\usepackage{multirow}

\usepackage[colorlinks,
linkcolor=red,
anchorcolor=blue,
citecolor=green
]{hyperref}
\usepackage
[
        letterpaper,
        left=1.91cm,
        right=1.91cm,
        top=1.91cm,
        bottom=2cm,
]
{geometry}

\newtheorem{assumption}{Assumption}
\newtheorem{proposition}{Proposition}

\newtheorem{remark}{Remark}

\newcommand{\x}{{\bf x}}

  \usepackage{pgfplots}
  \pgfplotsset{compat=newest}
  \usetikzlibrary{plotmarks}
  \usetikzlibrary{arrows.meta}
  \usepgfplotslibrary{patchplots}
  \usepackage{grffile}
  \usepackage{amsmath}

\begin{document}
\begin{frontmatter}
\title{{\bf Joint Infrastructure Planning and Order Assignment for On-Demand Food-Delivery Services with Coordinated Drones and Human Couriers}}
\author[hkust_address]{Yang Liu}
\ead{yliuiw@connect.ust.hk}
\author[hkust_address]{Yitong Shang}
\ead{ytshang@ust.hk}
\author[hkust_address]{Sen Li}
\ead{cesli@ust.hk}

\address[hkust_address]{Department of Civil and Environmental Engineering, The Hong Kong University of Science and Technology, \\Clear Water Bay, Kowloon, Hong Kong}
\begin{abstract}
This paper investigates the optimal infrastructure planning and order assignment problem of an on-demand food-delivery platform with a mixed fleet of drones and human couriers. The platform serves the orders with two delivery modes: (a) ground delivery and (b) drone-assisted delivery (i.e., air delivery). In ground delivery, an order is collected at the restaurant and transported to the destination by a human courier. For air delivery, the delivery process is segmented into three legs: initially, a human courier picks up the order from the restaurant and transports it to a nearby launchpad. These launchpads are staffed by personnel who load the orders onto drones and replace batteries as needed. The loaded drone then transports the order from the launchpad to a kiosk, which is an automated facility designed to accommodate drone landings and securely store orders. Subsequently, another courier retrieves the order from the kiosk for delivery to the final destination. The platform must determine the optimal locations for launchpads and kiosks within a transportation network, and devise an order assignment strategy that allocates food-delivery orders between ground and air delivery, while considering the bundling probabilities of ground deliveries and the waiting times at launchpads and kiosks for air deliveries. We formulate the platform's problem as a mixed-integer nonlinear program and develop a novel neural network-assisted optimization method to obtain high-quality solutions. The proposed model and algorithm are validated through a case study in Hong Kong, and the results reveal that the introduction of drone delivery will lead to reduced operational costs for the platform, a smaller courier fleet size, and increased opportunities for order bundling. Interestingly, we also find that the expansion of air delivery services may actually entail larger delivery times despite the significantly higher speeds of air compared to ground delivery. We attribute this phenomenon to the crucial trade-off between the travel time savings induced by the faster air delivery and the associated detours incurred by intermodal transfer and extra waiting times at launchpads and kiosks, which crucially depends on the distance of the orders and the sequence of activating long-distance air delivery routes versus short-distance ones. 
\end{abstract}
\end{frontmatter}

\section{Introduction}

On-demand food-delivery platforms such as UberEats, Deliveroo, DoorDash, and Meituan have become increasingly popular across various countries, revolutionizing the way we consume food. These platforms connect customers to a wide array of restaurants and couriers, enabling them to order their favorite meals online and have them delivered directly to their doorsteps. The COVID-19 pandemic further accelerated the growth of these platforms as customers became more accustomed to ordering meals online. In response to these changing consumer habits, restaurants have significantly adapted their business models. For instance, it is reported that in Hong Kong, restaurants experienced a 37\% increase in takeaway orders facilitated by food-delivery platforms in the first quarter of 2024. Moreover, 34\% of these restaurants indicated that delivery platforms contribute to at least 30\% of their revenue, underscoring the substantial impact these services have on the food industry.\footnote{Refer to https://hongkongbusiness.hk/food-beverage/news/hong-kong-restaurants-saw-37-uptick-in-takeaways-supported-food-delivery-aggregators.} 

Despite its great promise, the rapid growth of on-demand food-delivery platforms also presents significant societal challenges. Specifically, during peak hours, customers often experience delivery delays due to the high volume of orders, insufficient courier supply, and road congestion. To manage this surge in demand with a limited fleet size without compromising service quality, platforms frequently assign orders to couriers with very tight delivery windows, which can compromise courier safety. Consequently, couriers face immense pressure to race against the clock to avoid penalties for late deliveries, a situation that raises substantial social issues related to labor rights and road safety. This high-pressure environment amplifies the risk of traffic accidents, particularly during congested periods and adverse weather conditions.  In China, for instance, food delivery platforms have faced criticism for their unsympathetic   algorithms that impose strict delivery deadlines, placing excessive pressure on drivers and potentially increasing safety risks \cite{lai2020stuck}. The Shanghai Municipal People's Government reported that in the first half of 2019, there were 325 road traffic accidents involving the express delivery and food-delivery industries in the city, which resulted in 5 deaths and 324 injuries \cite{shanghai2020}.
These incidents illustrate the considerable pressure on stakeholders and highlight the urgent need to find lasting solutions that foster a healthy relationship between the platforms and gig workers, which are  crucial for the sustainable growth of the food-delivery market.

With the advancement of drone technology and the ongoing development of the low-altitude economy, drone delivery has emerged as a promising solution to the challenges faced by the current on-demand food-delivery market. Compared to human couriers, drones offer faster delivery speeds, lower operational costs, independence from traffic conditions, reduced environmental impacts, and increased accessibility, making them a viable alternative for delivering meals on demand. Many food-delivery companies worldwide have already endeavored to implement drone delivery services. For instance, Manna, an Irish drone delivery company, has launched takeaway deliveries in Dublin and plans to expand to 25 more locations \cite{curran2023}. Drone delivery company Wing launched drone-assisted food-delivery services in Melbourne in 2024, empowering 250,000 residents to order food through DoorDash with deliveries executed by small aircraft \cite{thorn2024wing}. Flytrex, a drone-based food delivery service operating in North Carolina and Texas, has announced that it successfully fulfilled 100,000 food delivery orders in 2024, reaching 70\% of households within its service areas \cite{flytrex2024}.
However, these initiatives primarily target rural and less populated urban areas, where suitable pick-up and drop-off locations can be easily found, allowing drones to complete the entire delivery process autonomously. Adapting this model to densely populated urban regions, where residents predominantly reside in apartment buildings, presents notable challenges as it is particularly difficult for drones to access the pickup points in restaurants located in large shopping malls and drop-off points in high-rise commercial or residential buildings. Yet, introducing drone-assisted food delivery services in these dense areas could be precisely where drones contribute the most. These locations typically experience higher order volumes, more severe traffic congestion, and less patient customers—all challenges that are difficult for human couriers to manage but can be effectively addressed by drone-assisted delivery. As such, drones could significantly enhance delivery efficiency and customer satisfaction in high-demand urban settings.

In contrast to the aforementioned companies, Meituan, China's largest on-demand delivery platform, has introduced a novel collaborative delivery model, utilizing a mixed fleet of drones and human couriers within a multimodal delivery network \cite{meituan2023MIT}. This model allows for two types of delivery: (a) ground delivery and (b) drone-assisted delivery, hereafter referred to as air delivery.
In ground delivery, an order is collected at the restaurant and transported to its destination using various ground transportation modes such as cars, motorcycles, or bicycles, all managed by a human courier. For air delivery, the process is segmented: initially, a human courier picks up the order from the restaurant (potentially bundled with other orders) and transports it to a nearby launchpad. These launchpads are staffed by personnel who load the orders onto drones and replace batteries as needed. The loaded drone then transports the order from the launchpad to a kiosk (without bundling), which is an automated facility designed to accommodate drone landings and securely store orders. Subsequently, another courier retrieves the order from the kiosk for delivery to the final destination. This innovative use of a mixed fleet effectively addresses the challenges associated with the first and last legs of delivery by integrating drones with human couriers, thus facilitating food delivery from door to door in densely populated areas. With over 30 drone routes launched, Meituan has successfully delivered more than 300,000 orders so far in Shenzhen \cite{meituan2024}, and plans are underway to expand into the low-altitude economy in Hong Kong \cite{meituan2024hk}, showcasing a significant step forward in the deployment of drones in on-demand food-delivery services in densely populated urban environments.

The emergence of novel business models involving coordinated drones and human couriers for on-demand food delivery introduces unique strategic and operational challenges that are relatively underexplored in existing literature. To bridge this gap, our study investigates a platform similar to Meituan's in Shenzhen, China, which constructs infrastructure such as launchpads and kiosks and integrates drone and human courier services for on-demand food delivery in densely populated areas. We propose a mathematical model to characterize the decision-making problem of the food-delivery platform as a Mixed-Integer Nonlinear Program (MINLP). Additionally, we derive a novel neural network-assisted method to obtain high-quality solutions to this complex problem. The major contributions of this paper can be summarized as follows:
\begin{itemize}
    \item We develop a mathematical framework that characterizes the joint infrastructure planning and order assignment strategies for the on-demand food-delivery platform using a mixed fleet of drones and human drivers. This framework addresses the challenges faced by the platform in determining the optimal locations for launchpads and kiosks within a transportation network under a constrained budget, and in devising an order assignment strategy that allocates food-delivery orders demand between ground and air delivery. These decisions hinge on the bundling probabilities of ground deliveries and the waiting times at launchpads and kiosks for air deliveries. To capture the former, we developed a steady-state equilibrium model that prescribes the matching process between couriers and orders. To characterize the latter, we formulated a double-ended queue model for the interactions at launchpads between orders and drones. The overall profit maximization problem is formulated as a Mixed-Integer Nonlinear Programming (MINLP) problem. {\em To the best of our knowledge, this is the first investigation into the joint infrastructure planning and order assignment for food-delivery platforms with drone-assisted delivery services.}
    \item We have developed a novel neural network-assisted method to obtain high-quality solutions to the complex MINLP for the food-delivery platform. Specifically, we first isolate the nonlinear components of the objective functions, which depends on the decision variables. The evaluation of these components requires solving a fixed point involving numerous nonlinear constraints. We approximate this function using a neural network configured with two-dimensional inputs and outputs. This neural network is then trained by sampling within the input space, solving the fixed point for each input, and subsequently obtaining labeled data for the outputs. Once trained, the neural network is integrated into the platform's optimization problem as constraints, which effectively transforms the MINLP into a Mixed-Integer Linear Programming (MILP) model. This transformation allows for the application of standard off-the-shelf algorithms, which facilitates efficient and tractable solution finding. The effectiveness of the proposed algorithm has been validated and demonstrates excellent performance: the neural network successfully approximates the nonlinear component with sufficient accuracy, and the globally optimal solution of the approximated problem can be easily derived. 

    \item We conducted a case study in Hong Kong based on real data to investigate how reductions in infrastructure costs and increases in food-delivery demand influence the planning and operational decisions of the platform. Specifically, our investigation reveals that as the infrastructure costs of launchpads and kiosks decrease, the food-delivery platform opts to expand air delivery services across the transportation network, initially activating air delivery routes for long distances and subsequently expanding them to shorter routes. This expansion leads to reduced operational costs for the platform, a smaller courier fleet size, and increased opportunities for order bundling. However, the expansion of air delivery services does not necessarily entail reduced delivery times for customers, despite the significantly higher speeds of air compared to ground delivery. This highlights a crucial trade-off between the travel time savings induced by the faster air delivery and the associated detours incurred by intermodal transfer and extra waiting times at launchpads and kiosks. We find that this trade-off appears to be heavily dependent on the distance of the order. For long-distance trips, the detours are minor compared to the substantial time savings from faster drone delivery. In contrast, for short-distance trips, the detours and additional waiting times can significantly extend the delivery time compared to ground delivery (which is already short due to the short distance), outweighing the time saved by the drones. Consequently, as the platform initially activates air delivery routes for long distances and subsequently expands them to shorter routes, the average delivery time first decreases (due to the opening of long-distance routes) and then increases (due to the opening of short-distance routes).
\end{itemize}

\section{Literature review}
This section reviews the existing literature related to our research focus.  To streamline our discussion, we classify the relevant literature into two distinct categories: (1) research centered on food-delivery services and (2) studies focusing on drone delivery services.  A comparison between our study and the existing literature is discussed, and the contribution of this work is highlighted. 

\subsection{Food-Delivery Services}
As food-delivery services have attracted extensive attention in recent years, there is a growing body of literature focusing on the operational strategies of the food-delivery services. One of the most popular topics is the meal delivery and routing problem introduced by \cite{reyes2018meal}, in which the authors formulate the meal delivery routing problem as an expanded form of Pickup and Delivery Problem (PDP), with an optimization framework focusing on the courier assignment problem.  Building upon this foundation, \cite{yildiz2019provably} have advanced an exact solution methodology for the meal delivery routing problem posited by \cite{reyes2018meal}, resulting in the generation of higher-quality solutions. \cite{steever2019dynamic} introduced a novel food delivery model where a customer can order from multiple restaurants in a single order, indicating that a split delivery policy is able to save delivery time, while the non-split delivery mode can reduce operation costs. Imbalanced distribution of food-delivery couriers across space and time due to varying demand is considered by \cite{xue2021optimization}, in which the authors proposed a two-stage programming model to optimize the allocation of rider resources in both spatial and temporal dimensions. While the literature mentioned above primarily focuses on a deterministic model assuming perfect demand information, the scope of the meal delivery and routing problem has also been extended to a stochastic setting, incorporating uncertainties related to order placement times and meal preparation durations \cite{ulmer2021restaurant,kancharla2024meal}.
A review of the existing literature and potential research questions on the operations of on-demand food-delivery services are summarized in \cite{mao2022demand}.

Another stream of research focuses on the planning strategies of the on-demand food-delivery platform \cite{bahrami2023three,chen2024courier}. With a focus on the management of courier supply,  \cite{yildiz2019service} developed a stylized equilibrium model to investigate the impacts of service coverage area and courier's delivery acceptance behaviors on the food-delivery platform profits, providing insights on the management of crowd-sourced delivery. A queuing-model-based framework was proposed by \cite{weng2021labor} to explore the courier assignment policy aiming at labor protection. Aside from traditional human-courier-based delivery services, the research on optimal strategic planning problems for food-delivery platforms is also extended to scenarios with the introduction of autonomous vehicles \cite{ye2022modeling}. On the customer demand side, \cite{du2023implications} investigated the impacts of compensation paid to customers by the platform or the restaurant for delivery delay, which is shown to be related to the level of service fee proportion. 
There is also a growing body of research focus on the management strategies on restaurants. For instance, \cite{feldman2023managing} developed a stylized model of a restaurant to explore the optimal platform-restaurant relationships for joint profit-maximization, indicating that an alternative contract where the platform pays the restaurant a revenue share and a fixed fee can lead to Pareto improvement compared with the traditional contract. Expanding on this, \cite{chen2022food} proposed a similar problem with two streams of customers (one has access to food-delivery services and one does not) using an unobservable queue model. Research by \cite{liu2023economic} explored the effects of capping commissions that online food-delivery platforms charge restaurant owners, demonstrating that an inappropriate commission cap could have negative consequences for all stakeholders within the food-delivery market. However, a common limitation of the aforementioned literature regarding planning strategies for the food-delivery platform is the focus on an aggregated level that overlooks the spatial dynamics.

While order bundling plays a crucial role in food-delivery services, the current body of research in this area remains limited. However, it is similar to the ride-pooling services, which have attracted extensive attention from the existing literature \cite{ke2020ride,zhang2022mitigating,zhang2021pool}. At the aggregate level without considering spatial heterogeneity, 
\cite{ke2020pricing} developed a mathematical model to evaluate the optimal management strategies of the ride-sourcing platform considering two markets: with and without ride-pooling services. It finds that under certain condition, the optimal ride fares in ride-pooling market is lower than that in the non-pooling market. This framework is further extended to incorporate the traffic congestion externality considering exogenous passenger demand \cite{ke2020ride} and elastic passenger demand \cite{ke2022coordinating}. \cite{bahrami2022optimal} evaluated the competitive advantages of solo ride-hailing services and pool ride-hailing services based on a closed queuing system under different scenarios regarding matching friction and marginal fleet acquisition cost. The findings revealed that it is optimal for the platform to offer both solo and pool services concurrently only when the matching is decreasing returns to scale. There is also an extensive body of literature on the planning strategies of ride-pooling services on the network level. For instance, \cite{wang2021predicting} developed a steady-state mathematical modeling framework to characterize the matching mechanism of ride-pooling services over a transportation network. Leveraging this model,  the authors predict the matching probabilities and expected delivery time for orders corresponding to each OD pair given an exogenous demand arrival rate. It is further be extended by \cite{wang2024optimizing,yang2024prediction}, where \cite{wang2024optimizing} investigated the impacts of OD-based discounting strategies on the ride-pooling system, and \cite{yang2024prediction} evaluates whether the prediction of matching opportunities established in \cite{wang2024optimizing} can improve the vehicle dispatching strategies for ride-pooling platforms. \cite{liu2023planning} presented a multi-zone queuing network model to capture the optimal planning strategies of ride-pooling services considering the spatial diversity of passenger demand. It offered valuable insights into the optimal fleet management strategies within ride-pooling systems.
While the bundling of orders in food-delivery services shares similarities with ride-pooling services, a prevalent assumption in most existing research on ride-pooling is that a vehicle can match with a maximum of two passengers. To overcome this limitation, \cite{ye2024modeling,ye2024modelingspatial} investigated the optimal pricing strategies for food-delivery services that utilize order bundling with flexible bundle sizes within an aggregated market equilibrium model and a network model, respectively. They provided managerial insights into the platform’s long-term strategic decisions.

{\em All of the aforementioned literature on food-delivery services in planning level only utilize ground transportation mode such as human couriers or autonomous vehicles. In contrast, our study complements existing research by introducing a multimodal food-delivery framework that integrates both ground transportation and low-altitude transportation using a mixed fleet of human couriers and drones. This unique perspective has not been explored in the existing works referenced earlier.}

\subsection{Drone-Assisted Delivery Services}
With the emergence and advancement of drone technology, the low-altitude economy has garnered significant attention. A substantial body of literature has been dedicated to exploring routing strategies for drone-assisted delivery services \cite{dorling2016vehicle, wen2022heterogeneous, huang2022solving}. One stream of research has focused on the coordinated management of trucks and drones for last-mile delivery services. For instance, \cite{murray2015flying} introduced the flying sidekick traveling salesman problem (FSTSP) to optimize demand assignment strategies for parcel delivery services using a hybrid fleet of trucks and drones, where drones can be transported by trucks due to battery limitations. Furthermore, \cite{agatz2018optimization} proposed a novel framework for the FSTSP formulated as an integer programming model and devised an algorithm with heuristics derived from local research, resulting in enhanced solution quality compared to \cite{murray2015flying}. Similarly, \cite{karak2019hybrid,moshref2020truck} concentrated on the development of efficient solution algorithms tailored with customized heuristics to address the joint vehicle-drone routing problem for pickup and delivery services. Other research on vehicle routing problem with drones focus on network planning and facility location \cite{pinto2020network,zhu2022two,bruni2023drone}, online vehicle dispatch strategies \cite{liu2019optimization}, the integration with public transit systems \cite{choudhury2021efficient,moadab2022drone}, and the utilization of machine learning-based method \cite{munoz2019deep,chen2022deep,cui2024dynamic}. For a comprehensive review of related literature on drone-assisted vehicle routing problem, please refer to \cite{macrina2020drone}.

Aside from operational strategies of drone-assisted delivery services, another stream of research investigates these services at the strategic level, focusing specifically on planning strategies and supply-demand analysis at the economic equilibrium. For instance, \cite{shavarani2019congested} provided an economic analysis of the drone delivery system with uncertainty. An M/G/K queueing system is developed to characterize the customer waiting time, and a mathematical model is formulated to determine the optimal facility location planning strategies with the objective of cost minimization. \cite{chen2021improved} investigated the pricing and fleet management strategies of drone-based delivery services, suggesting that fleet expansion is warranted in cases of high opportunity costs and short delivery times. Similarly, \cite{pei2021managing} also delved into the pricing and fleet size management strategies of drone delivery services, where a time-varying queuing model is developed to capture the system dynamics taking customer behaviors into account. The study's results indicated that while the demand fluctuated over time, the optimal pricing remained relatively stable. It was later extended by \cite{pei2023drone}, who examined a co-sourcing scenario involving drones and human couriers coexisting in a food delivery system. They found that the probability of outsourcing (delivery by human couriers) is influenced by both drone revenue and human courier earnings.

Our work significantly differs from all previously mentioned studies in that prior research on planning strategies and economic analysis of drone delivery services has predominantly focused on either drone-only delivery systems or co-sourcing models where drones and human couriers serve as substitutes rather than collaborators. {\em In contrast, our study explores the collaboration between drones and human couriers within a multimodal delivery network. This approach is a more promising strategy for using drones in food delivery as it overcomes the limitations of drones in addressing the first and last legs of the delivery process. It is also consistent with the current industrial practice of companies like Meituan. Moreover, the coordination of drones and ground delivery for a single order's delivery results in a more intricate demand assignment and fleet management strategies, as well as posing additional challenges regarding the planning of infrastructures that accommodate the departure and landing of drones. These perspectives have not been sufficiently examined in the existing literature.}

\section{The Mathematical Model}

We consider a transportation network represented by the graph $G(\mathcal{N}, \mathcal{A})$, where $\mathcal{N}$ denotes the set of nodes, and $\mathcal{A}$ corresponds to the set of links. A food-delivery platform offers on-demand services to customers, with each order originating from a node $i \in \mathcal{N}$ (i.e., a restaurant location) and destined for a node $j \in \mathcal{N}$. This platform utilizes a mixed fleet of drones and human couriers, coordinated to deliver meals within a multimodal delivery network. Orders can be delivered in two ways: (a) ground delivery and (b) drone-assisted delivery (hereafter referred to as air delivery). In ground delivery, an order is collected at node $i$ and delivered to node $j$ using ground transportation modes such as cars, motorcycles, or bicycles. For air delivery, the entire delivery process can be divided into several segments: initially, a human courier picks up the order from restaurant $i$ (potentially bundled with other orders) and transports it to a nearby launchpad $l$. The loaded drone then transport the order from launchpad $l$ to a kiosk $k$ (without bundling), which is 
an automated device to accommodate drone landing and order storage, while the drone will return to the launchpad immediately after dropping off the order. Subsequently, another courier picks up the order at kiosk $k$ for delivery to destination node $j$. An illustrative figure for launchpads and kiosks for Meituan is shown in Figure \ref{fig:meituan}, and the delivery processes for both ground and air delivery services are depicted in Figure \ref{fig:delivery_workflow}. 

In subsequent subsections, we will present a mathematical model that characterizes the demand allocation across different delivery modes, interactions among food-delivery orders, and the coordinated management of couriers and drones, taking into account factors such as launchpad congestion and ground delivery bundling probabilities. We will also address cost-minimization and waiting time-minimization strategies for the food-delivery platform under market equilibrium conditions. A summary of notations is presented in Appendix \ref{append:notation}.

\begin{figure}[htb!]
    \centering
    \includegraphics[width=1\linewidth]{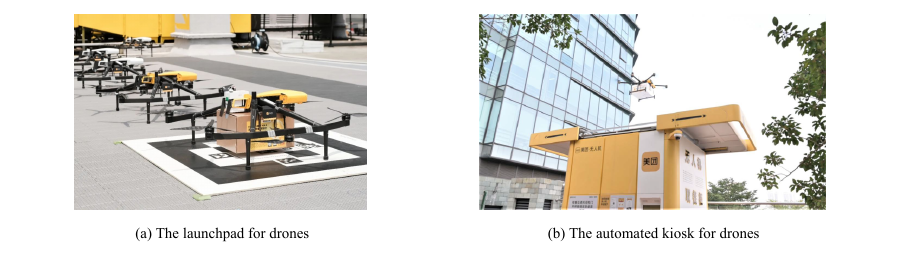}
    \caption{An illustrative figure for launchpads and kiosks for Meituan}
    \label{fig:meituan}
\end{figure}

\begin{figure}[htb!]
    \centering
    \includegraphics[width=1\linewidth]{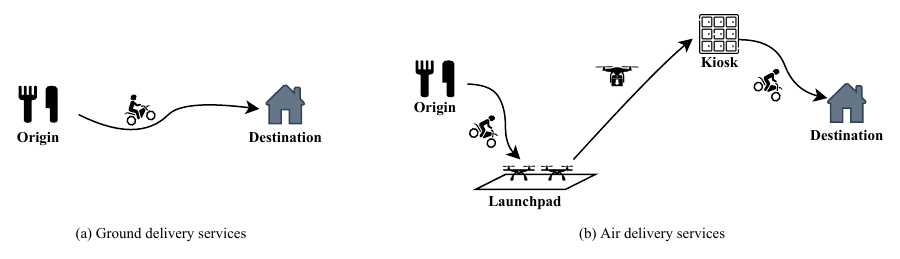}
    \caption{Delivery processes of ground delivery services and air delivery services}
    \label{fig:delivery_workflow}
\end{figure}


\subsection{Order Assignment}

After receiving an order, the platform determines whether to fulfill it using ground delivery exclusively or through a combination of drones and human couriers. In the latter case, the platform must also allocate the demand to launchpads near the origin and kiosks near the destination. To characterize these order assignment decisions, let $\mathcal{C}\subseteq\mathcal{N}$ represent the set of nodes corresponding to locations of customers (i.e., the destinations of the food-delivery orders), and denote $\mathcal{R}\subseteq\mathcal{N}$ denote the set of nodes corresponding to locations of restaurants (i.e., the origins of the food-delivery orders). It is important to note that $\mathcal{C}$ and $\mathcal{R}$ do not indicate individual customers or restaurants. Instead, they correspond to residential communities and commercial centers, respectively, where multiple customers or restaurants may co-exist and are densely clustered within a neighborhood. Let $\mathcal{L}\subseteq\mathcal{N}$ indicate the set of potential launchpad locations from which the platform can opt to build launchpads, and let $\mathcal{K}\subseteq\mathcal{N}$ identify the set of potential kiosk locations.\footnote{Note that the sets $\mathcal{C}$, $\mathcal{R}$, $\mathcal{L}$, $\mathcal{K}$ may intersect.} The platform may choose to construct launchpads at locations within $\mathcal{L}$ and establish kiosks at locations within $\mathcal{K}$. Let $y_l \in {0,1}$ ($l \in \mathcal{L}$) indicate the platform's decision to construct a launchpad at node $l$, with $y_l=1$ signifying the presence of a launchpad and $y_l=0$ indicating its absence. Similarly, let $z_k \in {0,1}$ ($k \in \mathcal{K}$) represent whether a kiosk is constructed at node $k$, where $z_k=1$ denotes its presence and $z_k=0$ its absence. 

The arrival of food-delivery orders from $i \in \mathcal{R}$ to $j \in \mathcal{C}$ is assumed to follow a Poisson process with an arrival rate $\lambda_{ij}$. Let $\lambda_{ilkj}^a$ denote the air delivery demand from origin $i$ to destination $j$, facilitated by a drone departing from launchpad $l$ and arriving at kiosk $k$. Let $\lambda_{ij}^g$ represent the demand for orders from origin $i$ to destination $j$ fulfilled by ground delivery services.
The assignment of food-delivery orders with origin $i$ and destination $j$ can be captured by the following relations:
\begin{align}
    &\lambda^a_{iljk},\ \lambda_{ij}^g \ge 0,\quad \forall i\in\mathcal{R}, j\in\mathcal{C}, l\in\mathcal{L} ,k\in\mathcal{K}\\
    &\sum_{i\in\mathcal{R}}\sum_{j\in\mathcal{C}}\sum_{k\in\mathcal{K}}\lambda^a_{ilkj}\le y_l M_L^l,\quad \forall l\in\mathcal{L} \label{neq:lambda_at_l}\\
    &\sum_{i\in\mathcal{R}}\sum_{j\in\mathcal{C}}\sum_{l\in\mathcal{L}}\lambda^a_{ilkj}\le z_{k}M_K^k,\quad \forall k\in\mathcal{K} \label{neq:lambda_at_k}\\
    &\sum_{l\in \mathcal{L}}\sum_{k\in\mathcal{K}} \lambda^a_{ilkj}+\lambda_{ij}^g\ = \lambda_{ij},\quad\forall i\in \mathcal{R}, j\in\mathcal{C} \label{eq:demand_allocation}
\end{align}
where $M_L^l$ and $M_K^k$ are sufficiently larger numbers. Inequalities (\ref{neq:lambda_at_l}) and (\ref{neq:lambda_at_k}) ensure that air delivery services can only be offered when there are launchpads and kiosks in corresponding locations, and (\ref{eq:demand_allocation}) guarantees that all the demand must be allocated to either ground delivery services or air delivery services.

Since air delivery involves ground segments in both the first and last mile, the overall ground flow includes contributions from both air and ground delivery flows. Let $\lambda_{ij}^{g,all}$ represent the total ground delivery flows. This encompasses ground-delivery orders originating from $i$ and destined for $j$, as well as air-delivery orders originating from $i$ and transported to launchpad $j$ (if $j \in \mathcal{L}$), and air-delivery orders destined for $j$ and transported from kiosk $i$ to $j$ (if $i \in \mathcal{K}$). In this case, we have the following relations:
\begin{align}
    \lambda_{ij}^{g,all} = 
        \lambda_{ij}^g+b_1\sum_{k\in\mathcal{K}} \sum_{j'\in\mathcal{C}}\lambda_{ijkj'}^a+b_2\sum_{i'\in\mathcal{R}}\sum_{l\in\mathcal{L}}\lambda_{i'lij}^a,
\end{align}
where
\begin{align}
    b_1 = \begin{cases}
        1, &i\in\mathcal{R},\ j\in\mathcal{L}\\
        0, &\text{otherwise}
    \end{cases},\quad
    b_2 = \begin{cases}
        1, & i\in\mathcal{K},\ j\in\mathcal{C}\\
        0, &\text{otherwise}
    \end{cases}.
\end{align}
Ground delivery flows may involve multiple orders consolidated within a single human courier, thus the detailed delivery time and operational costs depend on how orders are matched with couriers. Meanwhile, the launchpad may experience congestion due to imbalanced inflow and outflow, leading to orders waiting for available drones or idle drones waiting for incoming orders in cases of supply shortages or surpluses. The waiting time is contingent upon the matching process between orders and drones. Below, we will detail the matching processes between orders and human couriers, as well as between orders and drones, to quantify the associated delivery times and operational costs.

\subsection{Matching between Customers and Couriers}

The food-delivery platform considers order bundling, enabling drivers to carry multiple food-delivery orders simultaneously. While order bundling enhances resource efficiency and cost-effectiveness, it may result in additional detours and impact customer waiting times negatively.  In order to evaluate the effects of order bundling on the delivery time and courier fleet utilization, we develop the expected delivery time  for distinct delivery modes and the shared delivery time for orders delivered on ground by characterizing the matching process between the customers and the human couriers.

We consider a scenario in which orders originating from similar locations and destined for similar locations are bundled together, assigned to the same human courier. In particular, we assume that orders from the same origin $i\in \mathcal{R}$ and destination $j\in \mathcal{C}$ areas are bundled. It is important to clarify that this does not necessarily mean that orders must originate from the same restaurant or be delivered to the exact same customers. Rather, it allows for orders to be sourced from multiple distinct restaurants within the same commercial area or delivered to multiple customers within the same residential community.
We adhere to the ``first-pickup-first-dropoff" principle for order bundling, ensuring that orders are delivered in the same sequence they were picked up. Additionally, for simplicity, we assume that a courier can deliver a maximum of three orders at a time,\footnote{This assumptions can be relaxed, at the expense of a more complicated model. } and a courier matched with fewer than three orders may only be matched with additional orders when she is en route to pick up the orders. Once all orders are picked up and the delivery process has begun, the courier becomes unavailable for further order matching. An example of pickup and delivery process of a courier matched with two food-delivery orders is shown in Figure \ref{fig:bundling_two}.

\begin{figure}[htb!]
    \centering
    \includegraphics[width=0.8\linewidth]{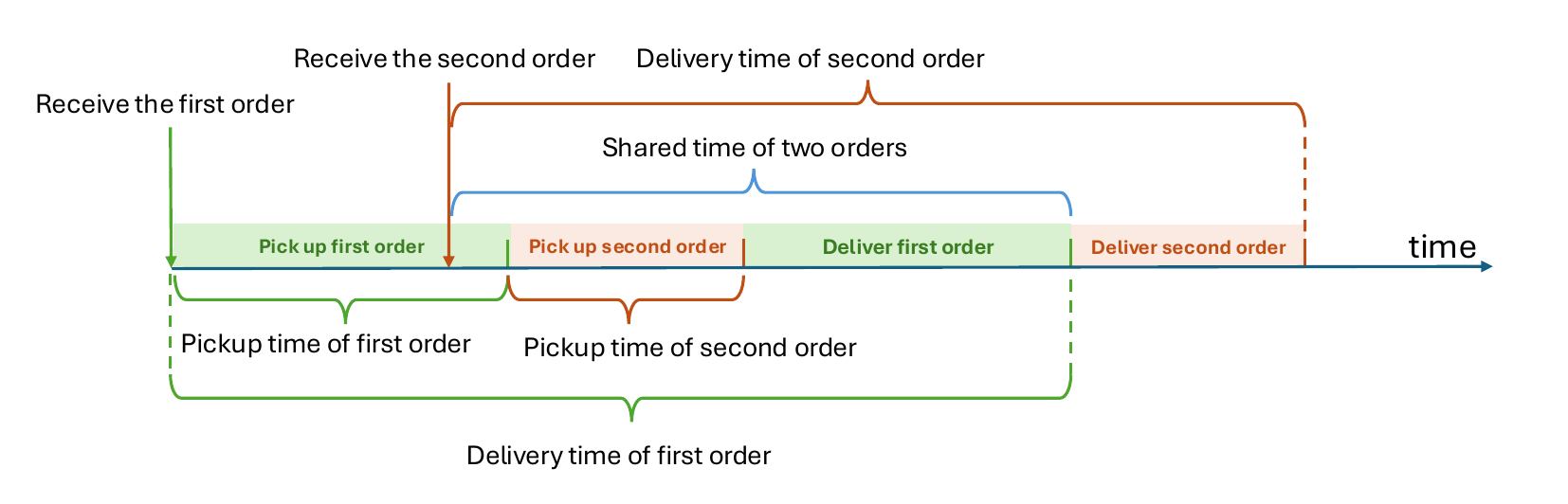}
    \caption{An example of pickup and delivery process of a courier matched with two food-delivery orders.}
    \label{fig:bundling_two}
\end{figure}

\begin{remark}
In this paper, we assume that only orders with similar origins and destinations can be bundled together. We believe this is a reasonable assumption as it helps to minimize detours and enhance delivery speeds, which are critical for the platform since food delivery customers are typically much more impatient than those in other sharing economies. This assumption aligns well with industry practices observed in Hong Kong (such as Foodpanda). The ``first-pickup-first-dropoff" principle, along with the assumption that a courier becomes unavailable for further order matching once delivery has commenced, naturally follows from our initial assumption that orders are bundled based on proximity of origins and destinations. This is because when orders are bundled in such a manner, it becomes inefficient for a courier already en route to reverse direction and return to the origin for additional pickups.
\end{remark}

We refer to the courier currently en-route to pick up a food-delivery order and available for new matches as a ``taker", represented by $t_{(i,j,n)}^m$, where $n$ represents the number of orders dispatched to the taker to be delivered from $i$ to $j$, and $m$ ($m\le n$) represents the current count of the order being picked up.  In the considered scenario, we have three taker states for the takers of each OD pair: $t_{(i,j,1)}^1$, $t_{(i,j,2)}^1$ and $t_{(i,j,2)}^2$. These represent, respectively, couriers who are dispatched with one order and are en route to pick up the first order; couriers who are dispatched with two orders and are en route to pick up the first order; and couriers who are dispatched with two orders and are en route to pick up the second order. It is important to note that the term ``taker" does not include ``idle couriers", who refer to couriers that are waiting and have not been dispatched to any order.  Furthermore, we define a ``seeker", denoted by $s_{(i,j)}$, as a food-delivery order from restaurant/kiosk $i$ to launchpad/customer location $j$ waiting to be matched with a taker or dispatched to an idle courier. We assume that if there exists an available taker, a seeker will first be matched to the taker.\footnote{Matching a seeker to a taker is more beneficial than to an idle driver, as this can enhance courier utilization and consequently lead to cost reduction.} In cases where no takers are available when the seeker places an order, the platform will assign the nearest idle driver to pick it up, and the idle courier then transitions into a taker in state $t_{(i,j,1)}^1$. When a taker in state $t_{(i,j,n)}^m$ receives a new food-delivery order before picking up the current order, it progresses into state $t_{(i,j,n+1)}^m$.


\subsubsection{Bundling Probabilities}\label{sec:bundle_prob}
We first investigate the probabilities for order bundling, which dictate the state transitions of seekers and takers, playing a fundamental role in the derivation of delivery times. Specifically, we calculate the bundling probabilities for orders in different taker states and seeker states.

Let $p_{t_{(i,j,n)}^m}$ be the probability for a taker in state $t_{(i,j,n)}^m$ being matched with a seeker. As taker can only be matched with a seeker while on the way to pick up the current order, the probability $p_{t_{(i,j,n)}^m}$ is contingent on the comparison between the courier's pickup time and the seeker's arrival time. 
Let $t_i^c$ be the average time for a courier to pick up the second/third order after picking up the first/second order, which is exogenous and determined by the density of the restaurants/kiosks at node $i$.
However, the average time for a courier to pick up the first order is more complex.
Since the platform will assign the seeker to the nearest idle courier if no taker is available, the time required for a courier to pick up the first order is reliant on the distribution of nearby idle drivers. 
To characterize this, we divide the transportation network into distinct zones denoted by the set $\mathcal{Z}$, each comprising multiple nodes. Each node $i$ has its corresponding zone, denoted by $Z(i)\in\mathcal{Z}$. Let $N_{Z(i)}^I$ be the average number of idle drivers in zone $Z(i)$.
Assuming that the idle drivers are uniformly and independently distributed across the zone, and that food-delivery orders originating from node $i$ can only be dispatched to the nearest idle driver within zone $Z(i)$ in the absence of an available taker, using the well-established ``square root law" \cite{arnott1996taxi,li2019regulating}, the time for the courier to pick up the first food-delivery order is given by:
\begin{align}\label{eq:pickup_time}
    w_i^c = \frac{H_i}{\sqrt{N_{Z(i)}^I}},
\end{align}
where $H_i$ is a parameter that  depends on the geometry of zone i, and the average speed of idle drivers.

Assume that the time for an idle courier to arrive at zone $i$, denoted by $W_i^c$, is exponentially distributed with expectation $w_i^c$. Assume that the time for a taker in state $t_{(i,j,n)}^m$ ($m>1$) to pick up the current order is $T_i^c$, which is exponentially distributed with expectation $t_i^c$. 
Let $T_{t_{(i,j,n)}^m}$ denote the waiting time for a matching opportunity to occur for takers in state $t_{(i,j,n)}^m$. Since the arrival of seekers is Poisson with rate $\lambda_{ij}^{g,all}$, $T_{t_{(i,j,n)}^m}$ is exponentially distributed with expectation $\frac{1}{\lambda_{ij}^{g,all}}$.
Thus, the probability of a taker in state $t_{(i,j,n)}^m$ to be matched with a seeker can be written as:
\begin{align}\label{eq:p_t}
    p_{t_{(i,j,n)}^m} =  \begin{cases}
        P\left(T_{t_{(i,j,n)}^m} < W_i^c\right) = \dfrac{\lambda_{ij}^{g,all}w_i^c}{1+\lambda_{ij}^{g,all}w_i^c}, & m = 1 \\
        P\left(T_{t_{(i,j,n)}^m} < T_i^c\right) = \dfrac{\lambda_{ij}^{g,all}t_i^c}{1+\lambda_{ij}^{g,all}t_i^c}, & m>1
    \end{cases}.
\end{align}

Next, we will characterize $p_{s_{(i,j)}}$, which represents the probability for a seeker $s_{(i,j)}$ to be matched with a taker.  Given that a taker is prioritized over an idle driver in matching with a seeker, each seeker is always assigned to the taker if one is available. Therefore, $p_{s_{(i,j)}}$ can be derived as the probability that a taker for the OD pair $ij$ exists. If we denote $\rho_{t_{(i,j,n)}^m}$ as the probability that a taker in state $t_{(i,j,n)}^m$ exists, then we have:
\begin{align} \label{eq:p_s}
    p_{s_{(i,j)}} = \rho_{t_{(i,j,1)}^1}+\rho_{t_{(i,j,2)}^1}+\rho_{t_{(i,j,2)}^2}.
\end{align}
Equation (\ref{eq:p_s}) specifies that a seeker can only be matched with a taker if there is an available taker who shares the same origin and destination as the seeker. Such a taker may potentially be in one of the following three statuses: $t_{(i,j,1)}^1$, $t_{(i,j,2)}^1$, or $t_{(i,j,2)}^2$.

To further determine $\rho_{t_{(i,j,n)}^m}$, we need to calculate the expected time that a taker in state $t_{(i,j,n)}^m$ is available for matching, denoted by $\tau_{t_{(i,j,n)}^m}$.
A taker is eligible for pairing with a seeker if: (1) the courier is on the way to pick up the order; and (2) the number of orders dispatched to the taker is less than its capacity. 
Let \(T_{t_{(i,j,n)}^m}^{\text{match}}\) represent the time a taker in state \(t_{(i,j,n)}^m\) is available for matching, which is a random variable with expectation \(\tau_{t_{(i,j,n)}^m}\). Specifically, when \(m = 1\), the average time a taker in state \(t_{(i,j,n)}^1\) is available is the minimum value between the time it takes the courier to pick up the first order and the time from when the taker changes to state \(t_{(i,j,n)}^1\) until a new food-delivery order is dispatched. Therefore, we have $T_{t_{(i,j,n)}^m}^{\text{match}} = \min\{T_{t_{(i,j,n)}^m}, W_i^c\}$, which is subject to an exponential distribution with rate \(\lambda_{ij}^{g,\text{all}} + \frac{1}{w_i^c}\). Similarly, when \(m = 2\), we have \( T_{t_{(i,j,n)}^m}^{\text{match}} = \min\{T_{t_{(i,j,n)}^m}, T_i^c\}
\) following an exponential distribution with rate \(\lambda_{ij}^{g,\text{all}} + \frac{1}{t_i^c}\). Overall, the expectation term \(\tau_{t_{(i,j,n)}^m}\) can be derived as:
\begin{align}
    \tau_{t_{(i,j,n)}^m} = \begin{cases}
        \frac{1}{\lambda_{ij}^{g,all}+\frac{1}{w_i^c}} = \frac{w_i^c}{1+\lambda_{ij}^{g,all}w_i^c}, & m = 1\\
        \frac{1}{\lambda_{ij}^{g,all}+\frac{1}{t_i^c}} = \frac{t_i^c}{1+\lambda_{ij}^{g,all}t_i^c}, & m = 2
    \end{cases}
\end{align}

Given \(\tau_{t_{(i,j,n)}^m}\), we further need to characterize the arrival rate of unmatched taker in each state to derive  $\rho_{t_{(i,j,n)}^m}$ (detailed derivation will be discussed shortly). To this end, let $\lambda_{t_{(i,j,n)}^m}^{um}$ be the arrival rate of unmatched takers in state $t_{(i,j,n)}^m$, then the following relations hold:
\begin{align}
    & \lambda_{t_{(i,j,1)}^1}^{um} = \lambda_{ij}^{g,all}(1-p_{s_{(i,j)}}),\label{equation_11}\\
    & \lambda_{t_{(i,j,2)}^1}^{um} = \lambda_{t_{(i,j,1)}^1}^{um}p_{t_{(i,j,1)}^1} = \lambda_{ij}^{g,all}(1-p_{s_{(i,j)}})p_{t_{(i,j,1)}^1},\label{equation_12}\\
    & \lambda_{t_{(i,j,2)}^2}^{um} = \lambda_{t_{(i,j,2)}^1}^{um}(1-p_{t_{(i,j,2)}^1}) = \lambda_{ij}^{g,all}(1-p_{s_{(i,j)}})p_{t_{(i,j,1)}^1}(1-p_{t_{(i,j,2)}^1}).\label{equation_13}
\end{align}
where (\ref{equation_11}) indicates that a seeker $s_{(i,j)}$ transitions into the taker state $t_{(i,j,1)}^1$ if it is dispatched to an idle driver, which only occur when there is no other available taker, and thus it has a probability of $1-p_{s(i,j)}$; Equation (\ref{equation_12}) indicates that a taker in state $t_{(i,j,1)}^1$ evolves into $t_{(i,j,2)}^1$ when matched to a new seeker while on the way to pick up the first order, which has a probability of $p_{t_{(i,j,1)}^1}$; Equation (\ref{equation_13}) indicates that a taker in state $t_{(i,j,2)}^2$ can shift to $t_{(i,j,2)}^2$ if and only if it fails to be matched to new orders prior to picking up the first one, which has a probability of $1-p_{t_{(i,j,2)}^1}$. {For reader's convenience, the state transition across distinct seeker and taker states that describes the aforementioned relations are illustrated in Figure \ref{fig:taker_state_transition}.}
\begin{figure}[H]
    \centering
    \includegraphics[width=0.7\linewidth]{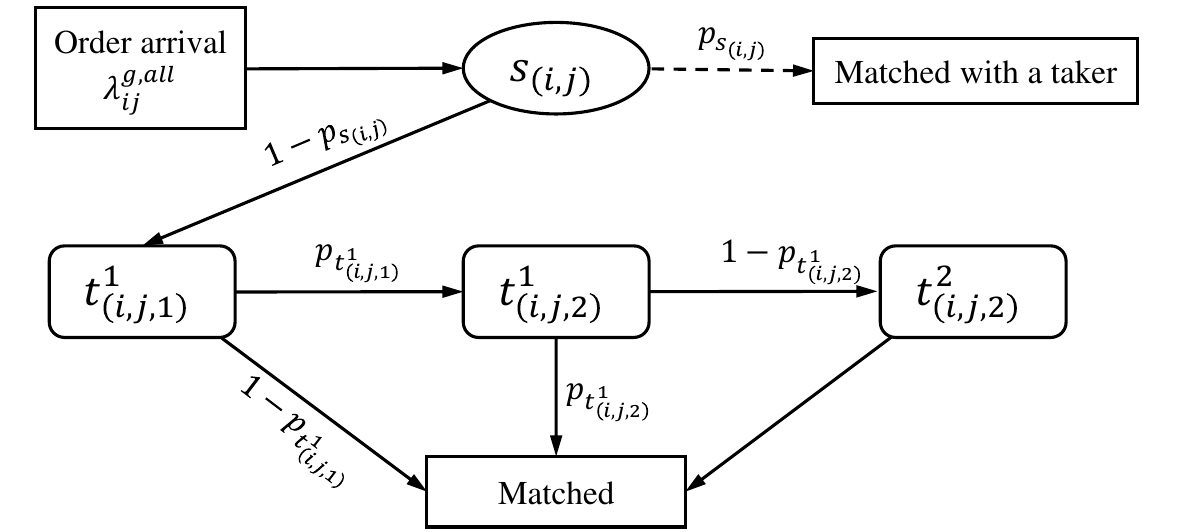}
    \caption{The state transitions across seeker and taker states and the corresponding probabilities (The ``Matched'' state signifies that the taker has completed matching with all seekers).}
    \label{fig:taker_state_transition}
\end{figure}

With the arrival rates of takers $\lambda_{t_{(i,j,n)}^m}^{um}$ and their average dwell time $t_{(i,j,n)}^m$, the probability of the existence of a taker in state  $t_{(i,j,n)}^m$ can be derived accordingly, based on the following proposition:
\begin{proposition}
    \label{prop_existence_taker}
    Assume that a seeker will be matched to a taker whenever one is available, and will only be matched to an idle courier if a taker does not exist. Additionally, assume that distinct orders do not arrive simultaneously.\footnote{This is a rather weak assumption since the chance of two orders arriving exactly at the same time can be is very small, and thus can be neglected without influencing the result.} Under these assumptions, the probability that a taker is in state $t_{(i,j,n)}^m$, denoted by $\rho_{t_{(i,j,n)}^m}$, can be derived as:
\begin{align} \label{eq:prob_exist_taker_11_eq}
    &\rho_{t_{(i,j,1)}^1} = \lambda_{t_{(i,j,1)}^1}^{um}\tau_{t_{(i,j,1)}^1} = \lambda_{ij}^{g,all}(1-p_{s_{(i,j)}})\frac{w_i^c}{1+\lambda_{ij}w_i^c},\\ \label{eq:prob_exist_taker_21_eq}
    &\rho_{t_{(i,j,2)}^1} = \lambda_{t_{(i,j,2)}^1}^{um}\tau_{t_{(i,j,2)}^1} = \lambda_{ij}^{g,all}(1-p_{s_{(i,j)}})\frac{\lambda_{ij}^{g,all}w_i^c}{1+\lambda_{ij}^{g,all}w_i^c}\frac{w_i^c}{1+\lambda_{ij}w_i^c},\\ \label{eq:prob_exist_taker_22_eq}
    &\rho_{t_{(i,j,2)}^2} = \lambda_{t_{(i,j,2)}^2}^{um}\tau_{t_{(i,j,2)}^2} = \lambda_{ij}^{g,all}(1-p_{s_{(i,j)}})\frac{\lambda_{ij}^{g,all}w_i^c}{\left(1+\lambda_{ij}^{g,all}w_i^c\right)^2}\frac{t_i^c}{1+\lambda_{ij}^{g,all}t_i^c}.
\end{align}
\end{proposition}
\begin{proof}
To prove equations (\ref{eq:prob_exist_taker_11_eq})-(\ref{eq:prob_exist_taker_22_eq}), we begin by acknowledging that under the stipulated assumptions, at most one taker can exist in each state $t_{(i,j,n)}^m$ at any given time. This limitation arises because a new taker is only generated when a seeker is matched with an idle courier. Once the first taker enters the system, any subsequent seeker will be matched to this existing taker with priority. Consequently, the creation of a second taker would require a seeker to be matched with another idle courier, which contradicts the initial assumption that seekers are matched to existing takers as a priority. This ensures that the presence of more than one taker in the same state at the same time is not possible. On the other hand, based on Little's law, $\lambda_{t_{(i,j,n)}^m}^{um}\tau_{t_{(i,j,n)}^m}$ represents the average number of takers in state $t_{(i,j,n)}^m$. Since there is either one taker or no taker in each state, $\lambda_{t_{(i,j,n)}^m}^{um}\tau_{t_{(i,j,n)}^m}$ is a real number between 0 and 1, which corresponds to the probability of existence of a taker in state $t_{(i,j,n)}^m$. This completes the proof. 
\end{proof}

By combining (\ref{eq:prob_exist_taker_11_eq})-(\ref{eq:prob_exist_taker_22_eq}) with (\ref{eq:p_s}), we can derive the following expression for  $\rho_{t_{(i,j,n)}^m}$:
\begin{align} \label{eq:prob_exist_taker_11}
    &\rho_{t_{(i,j,1)}^1} = \frac{\lambda_{ij}^{g,all}w_i^c}{1+\lambda_{ij}^{g,all}w_i^c\left(2+\frac{\lambda_{ij}^{g,all}w_i^c}{1+\lambda_{ij}^{g,all}w_i^c}+\frac{\lambda_{ij}^{g,all}t_i^c}{(1+\lambda_{ij}^{g,all}w_i^c)(1+\lambda_{ij}^{g,all}t_i^c)}\right)}
    \\ \label{eq:prob_exist_taker_21}
    &\rho_{t_{(i,j,2)}^1} = \rho_{t_{(i,j,1)}^1}\frac{\lambda_{ij}^{g,all}w_i^c}{1+\lambda_{ij}^{g,all}w_i^c},\\ \label{eq:prob_exist_taker_22}
    &\rho_{t_{(i,j,2)}^2} = \rho_{t_{(i,j,1)}^1}\frac{\lambda_{ij}^{g,all}t_i^c}{(1+\lambda_{ij}^{g,all}w_i^c)(1+\lambda_{ij}^{g,all}t_i^c)}.
\end{align}

Finally,  the bundling probabilities for seeker, denoted as $p_{s_{(i,j)}}$, can be obtained based on (\ref{eq:p_s}) under the given $\rho_{t_{(i,j,n)}^m}$ derived from (\ref{eq:prob_exist_taker_11})-(\ref{eq:prob_exist_taker_22}).

\subsubsection{Expected Delivery Time and Shared Delivery Time}


Given the bundling probabilities $p_{s_{(i,j)}}$ and $p_{t_{(i,j,n)}^m}$, we can calculate the expected delivery time and shared delivery time for bundled food-delivery services. For simplicity, we assume that the couriers are moving at a constant speed $v_c$. Let $L_{ij}^0$ denote the delivery distance from location $i$ to location $j$ without order bundling, then the time for a courier to travel from $i$ to $j$ is $t_{ij} = \frac{L_{ij}^0}{v_c}$. We also assume that the time between the dropoff of two bundled orders at node $i$ is also exponentially distributed with expectation $t_i^c$, which is the same as the time between the pickup of two bundled orders at node $i$.

To calculate the expected delivery time for a food-delivery order from location $i$ to location $j$, we will consider two distinct cases. The first cases is that the order is matched to a taker when it is at a seeker state $s_{(i,j)}$. In this case, there are three further scenarios, including (a) the order is matched to a taker in state $t_{(i,j,1)}^1$; (b) the order is matched to a taker in state $t_{(i,j,2)}^1$; and (c) the order is matched to a taker in state $t_{(i,j,2)}^2$. If no takers in these states are available, the order will then be matched to an idle courier and the courier turns into a taker, which constitutes the second case. Below we will characterize the expected delivery time for each case.  

\begin{proposition}
    \label{proposition2_case1_scenario1}
    If the order is matched to a taker $t_{(i,j,1)}^1$ while it is at a seeker state $s_{(i,j)}$, the expected delivery time for a seeker with origin $i$ and destination $j$(i.e., the time from when the seeker order is placed  to when it is delivered to the location $j$) can be expressed as:
\begin{align} \nonumber
    t_{s_{(i,j)}}^{t_{(i,j,1)}^1} = &p_{t_{(i,j,2)}^1}(\tau_{t_{(i,j,2)}^1}+w_i^c+2t_i^c)+(1-p_{t_{(i,j,2)}^1})p_{t_{(i,j,2)}^2}(\tau_{t_{(i,j,2)}^1}+\tau_{t_{(i,j,2)}^2}+2t_i^c)+\\
    &(1-p_{t_{(i,j,2)}^1})(1-p_{t_{(i,j,2)}^2})(\tau_{t_{(i,j,2)}^1}+\tau_{t_{(i,j,2)}^2})+t_{ij}+t_j^c,
\end{align}
\end{proposition}

\begin{proof}
If the order is matched to a taker $t_{(i,j,1)}^1$ while it is at a seeker state $s_{(i,j)}$, the taker then transitions into state $t_{(i,j,2)}^1$. In this scenario, the taker is matched to two orders, and will potentially have the chance of being matched to the third order. The movement of this taker, from the moment she receives the seeker's order (i.e., second order) to the moment she picks up all orders and begins the delivery, can be characterized by a state transition model illustrated in Figure \ref{fig:time_seeker}. The state $(i,j,k_1,k_2)$ in Figure \ref{fig:time_seeker} represents the status when the courier gets matched with $k_1$ orders with origin $i$ and destination $j$ and is currently on the way to pick up the $k_2$th order ($1\le k_2\le k_1$). The state ``deliver''  in Figure \ref{fig:time_seeker} signifies that the courier have picked up all bundled orders and begin delivery.
\begin{figure}[H]
    \centering
    \includegraphics[width=0.5\linewidth]{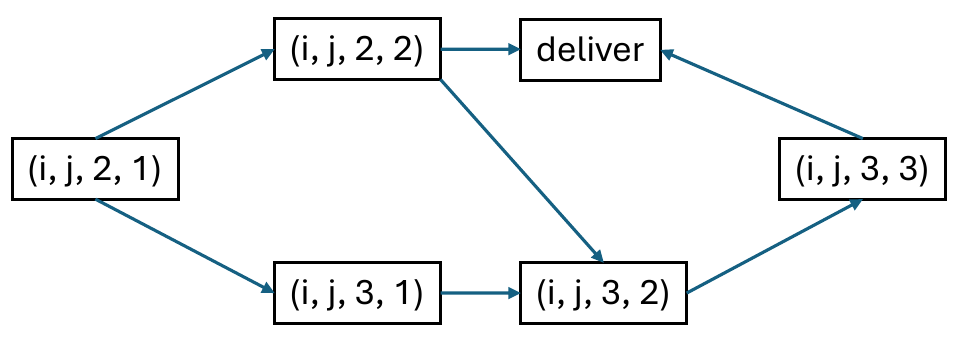}
    \caption{Movement of courier in taker state $t_{(i,j,1)}^1$ from getting matched with the seeker to picking up all orders.}
    \label{fig:time_seeker}
\end{figure}
Figure \ref{fig:time_seeker} illustrates three possible cases of the courier's movement:
\begin{enumerate}
\item $(i,j,2,1)\rightarrow (i,j,3,1) \rightarrow (i,j,3,2) \rightarrow (i,j,3,3)\rightarrow \text{deliver}$:
after getting matched to the seeker (i.e., second order), the taker receives the third order when she is on the way to pick up the first order.
    \item $(i,j,2,1)\rightarrow (i,j,2,2) \rightarrow (i,j,3,2) \rightarrow (i,j,3,3)\rightarrow \text{deliver}$: after getting matched to the seeker, the taker receives the third order when she is on the way to pick up the seeker (i..e, second order).   
    \item $(i,j,2,1)\rightarrow (i,j,2,2) \rightarrow \text{deliver}$: the taker is not matched to other orders after being matched to the seeker.    
\end{enumerate}
To derive the expected delivery time for each case described, we can model the process shown in Figure \ref{fig:time_seeker} as a Markov chain, where each path corresponds to one of the cases. The total expected delivery time for each case is calculated as the sum of the time spent at each node along the corresponding path. Specifically, the dwell time at each node, represented as $(i,j,k_1,k_2)$, can be denoted by $d_{(i,j,k_1,k_2)}$, then we have:
\begin{align}
\label{dwell_time}
    d_{(i,j,k_1,k_2)} = \begin{cases}
        \tau_{t_{(i,j,k_1)}^{k_2}}, & (k_1,k_2) \in \{(1,1),(2,1),(2,2)\}\\
        w_i^c, & k_1 = 3,k_2 = 1\\
        t_i^c, & k_1 =3, k_2 \ge 2
    \end{cases}
\end{align}
Note that case (i) occurs when the taker is matched to the third order before picking up the first order, which has a probability of $p_{t_{(i,j,2)}^1}$. In this case, the total delivery time is
\begin{equation}
d(i,j,2,1)+d(i,j,3,1)+d(i,j,3,2)+d(i,j,3,3)=\tau_{t_{(i,j,2)}^1}+w_i^c+2t_i^c.    
\end{equation}
Case (ii) arises when the taker is matched with a third order after the first order has been picked up but before the second order is collected. This scenario occurs with a probability of $(1-p_{t_{(i,j,2)}^1})p_{t_{(i,j,2)}^2}$, and the total delivery time is:
\begin{equation}
d(i,j,2,1)+d(i,j,2,2)+d(i,j,3,2)+d(i,j,3,3)=\tau_{t_{(i,j,2)}^1}+\tau_{t_{(i,j,2)}^2}+2t_i^c.    
\end{equation}
Case (iii) arises when the taker is not matched with a third order after the first order and second order has been picked up, which occurs with a probability of $(1-p_{t_{(i,j,2)}^1})(1-p_{t_{(i,j,2)}^2})$, and the total delivery time is:
\begin{equation}
d(i,j,2,1)+d(i,j,2,2)+d(i,j,3,2)+d(i,j,3,3)=\tau_{t_{(i,j,2)}^1}+\tau_{t_{(i,j,2)}^2}.    
\end{equation}
Based on the above discussion, the expected delivery time for a seeker with origin $i$ and destination $j$, matched to a taker in state $t_{(i,j,1)}^1$ (i.e., the time from when the seeker order is placed  to when it is delivered to the location $j$) is:
\begin{align} \nonumber
    t_{s_{(i,j)}}^{t_{(i,j,1)}^1} = &p_{t_{(i,j,2)}^1}(\tau_{t_{(i,j,2)}^1}+w_i^c+2t_i^c)+(1-p_{t_{(i,j,2)}^1})p_{t_{(i,j,2)}^2}(\tau_{t_{(i,j,2)}^1}+\tau_{t_{(i,j,2)}^2}+2t_i^c)+\\
    &(1-p_{t_{(i,j,2)}^1})(1-p_{t_{(i,j,2)}^2})(\tau_{t_{(i,j,2)}^1}+\tau_{t_{(i,j,2)}^2})+t_{ij}+t_j^c,
\end{align}
where the first three terms capture the average duration from when a seeker is matched with a taker until the taker collects all orders corresponding to cases (i)-(iii) in Figure \ref{fig:time_seeker}, $t_{ij}$ captures the time from when the taker collects all orders in $i$ to dropping off the first order in $j$, and $t_j^c$ captures the average time from the courier dropping off the first order to dropping off the seeker order. 
\end{proof}

\begin{proposition}\label{prop_t21}
If the order is matched to taker $t_{(i,j,2)}^1$ while it is at a seeker state $s_{(i,j)}$, the expected delivery time can be derived as:
\begin{align}
\label{time_prop3}
    t_{s_{(i,j)}}^{t_{(i,j,2)}^1} =w_i^c+2t_i^c+t_{ij}+2t_j^c,
\end{align}
\end{proposition}
\begin{proof}
If the order is matched to taker $t_{(i,j,2)}^1$, then the taker turns into state  $t_{(i,j,3)}^1$ and thus can not be matched to new orders since it already reaches the capacity. In this case, the seeker must be the last to be delivered. The first two terms in (\ref{time_prop3}) capture the average time from when the seeker is matched to the taker until the taker picks up all three orders. The last two terms represent the average time from when the taker collects all three orders until the final delivery of the seeker's order. This competes the proof. 
\end{proof}

\begin{proposition}\label{prop_t22}
If the order is matched to taker $t_{(i,j,2)}^2$, then the expected delivery time can be expressed as:
\begin{align}
    t_{s_{(i,j)}}^{t_{(i,j,2)}^2} =2t_i^c+t_{ij}+2t_j^c.
\end{align}
\end{proposition}
The proof of Proposition \ref{prop_t22} is similar to that of Proposition \ref{prop_t21} and is thus omitted. 

Combining the results of Propositions \ref{proposition2_case1_scenario1}-\ref{prop_t22}, the expected delivery time for a seeker in state $s_{(i,j)}$ getting matched with a taker can be characterized by:
\begin{align}
\label{t_Taker}
    t_{s_{(i,j)}}^{taker} = \frac{\rho_{t_{(i,j,1)}^1}t_{s_{(i,j)}}^{t_{(i,j,1)}^1}}{\sum_{t_{(i,j,m)}^n}\rho_{t_{(i,j,m)}^n}} +\frac{\rho_{t_{(i,j,2)}^1}t_{s_{(i,j)}}^{t_{(i,j,2)}^1}}{\sum_{t_{(i,j,m)}^n}\rho_{t_{(i,j,m)}^n}} +\frac{\rho_{t_{(i,j,2)}^2}t_{s_{(i,j)}}^{t_{(i,j,2)}^2}}{\sum_{t_{(i,j,m)}^n}\rho_{t_{(i,j,m)}^n}} 
\end{align}

While equation (\ref{t_Taker}) provides the expected delivery time for orders matched to a taker, we also need to derive the expected delivery time for orders matched to an idle courier. In this case, we have the following proposition:
\begin{proposition}
\label{prop_idle}
If a seeker in state $s_{(i,j)}$ is matched to an idle courier, then the expected delivery time for the order can be derived as:
\begin{align}\nonumber
    t_{s_{(i,j)}}^{idle} = &p_{t_{(i,j,1)}^1}p_{t_{(i,j,2)}^1}(\tau_{t_{(i,j,1)}^1}+\tau_{t_{(i,j,2)}^1}+w_i^c+2t_i^c)+p_{t_{(i,j,1)}^1}(1-p_{t_{(i,j,2)}^1})p_{t_{(i,j,2)}^2}(\tau_{t_{(i,j,1)}^1}+\tau_{t_{(i,j,2)}^1}+\tau_{t_{(i,j,2)}^2}+2t_i^c)\\\nonumber
    &+p_{t_{(i,j,1)}^1}(1-p_{t_{(i,j,2)}^1})(1-p_{t_{(i,j,2)}^2})(\tau_{t_{(i,j,1)}^1}+\tau_{t_{(i,j,2)}^1}+\tau_{t_{(i,j,2)}^2})+
    (1-p_{t_{(i,j,1)}^1})\tau_{t_{(i,j,1)}^1} +t_{ij},
\end{align}
\end{proposition}
\begin{proof}
When a seeker in state $s_{(i,j)}$ is matched to an idle courier, it would be the first order matched to this courier, and the courier then turns to a taker in state $t_{(i,j,1)}^1$. 
The movement of the taker from when she received the seeker's order to picking up all orders can be characterized by Figure \ref{fig:time_taker}.
\begin{figure}[H]
    \centering
    \includegraphics[width=0.64\linewidth]{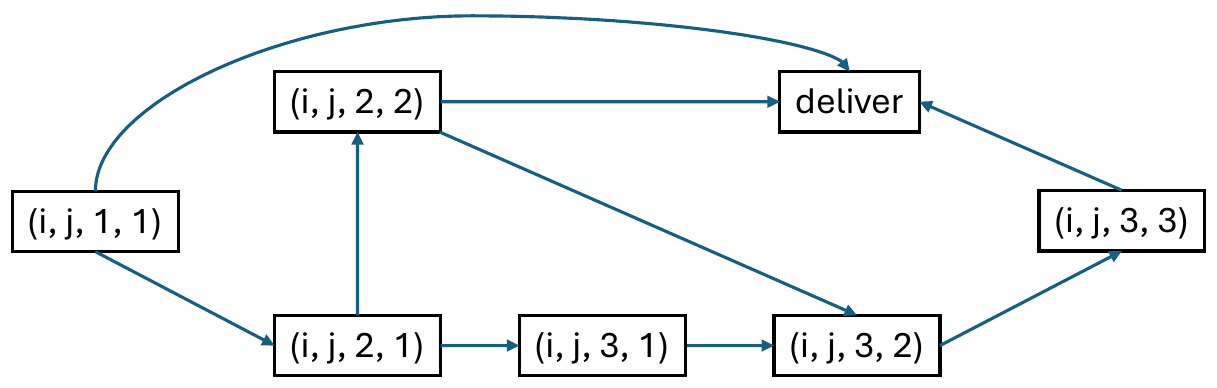}
    \caption{Movement of an idle courier from getting matched with the seeker to picking up all orders.}
    \label{fig:time_taker}
\end{figure}
Figure \ref{fig:time_taker} illustrates four possible cases of the courier's movement:
\begin{enumerate}
    \item $(i,j,1,1)\rightarrow (i,j,2,1) \rightarrow (i,j,3,1) \rightarrow (i,j,3,2)\rightarrow (i,j,3,3)\rightarrow \text{deliver}$: the courier received another two orders on the way to pick up the seeker.
    \item $(i,j,1,1)\rightarrow (i,j,2,1) \rightarrow (i,j,2,2) \rightarrow (i,j,3,2)\rightarrow (i,j,3,3)\rightarrow \text{deliver}$: the courier received the second order on the way to pick up the seeker (i.e., the first order), and the third one on the way to pick up the second order.
    \item $(i,j,1,1)\rightarrow (i,j,2,1) \rightarrow (i,j,2,2) \rightarrow \text{deliver}$: the courier received only one more order when she is on the way to pick up the seeker, and then begin to deliver.
    \item $(i,j,1,1) \rightarrow \text{deliver}$: the courier is not matched to other orders after being matched to the seeker.    
\end{enumerate}
Similar to Proposition \ref{proposition2_case1_scenario1}, to derive the expected delivery time for each case described, we can model the process shown in Figure \ref{fig:time_taker} as a Markov chain, where each path corresponds to one of the cases. The total expected delivery time for each case is calculated as the sum of the time spent at each node along the corresponding path. In this case, the dwell time at each node, denoted as $d_{(i,j,k_1,k_2)}$, satisfies (\ref{dwell_time}).

Case (i) occurs when the taker is matched to the second and third order before picking up the seekers order (i.e., first order), which has a probability of $p_{t_{(i,j,1)}^1}p_{t_{(i,j,2)}^1}$. In this case, the total delivery time is
\begin{equation}
d(i,j,1,1)+d(i,j,2,1)+d(i,j,3,1)+d(i,j,3,2)+d(i,j,3,3)=\tau_{t_{(i,j,1)}^1}+\tau_{t_{(i,j,2)}^1}+w_i^c+2t_i^c.    
\end{equation}
Case (ii) arises when the courier received the
second order on the way to pick up the seeker (i.e., the first order), and the third one on the way to
pick up the second order. This scenario occurs with a probability of $p_{t_{(i,j,1)}^1}(1-p_{t_{(i,j,2)}^1})p_{t_{(i,j,2)}^2}$, and the total delivery time is:
\begin{equation}
d(i,j,1,1)+d(i,j,2,1)+d(i,j,2,2)+d(i,j,3,2)+d(i,j,3,3)=\tau_{t_{(i,j,1)}^1}+\tau_{t_{(i,j,2)}^1}+\tau_{t_{(i,j,2)}^2}+2t_i^c.    
\end{equation}
Case (iii) arises when the courier received only one more order when she is on the way to pick up the seeker, which occurs with a probability of $p_{t_{(i,j,1)}^1}(1-p_{t_{(i,j,2)}^1})(1-p_{t_{(i,j,2)}^2})$, and the total delivery time is:
\begin{equation}
d(i,j,1,1)+d(i,j,2,1)+d(i,j,2,2)+d(i,j,3,2)+d(i,j,3,3)=\tau_{t_{(i,j,1)}^1}+\tau_{t_{(i,j,2)}^1}+\tau_{t_{(i,j,2)}^2}.    
\end{equation}
Based on the above discussion, the expected delivery time for a seeker who is matched to an idle driver can be derived as:
\begin{align}\nonumber
    t_{s_{(i,j)}}^{idle} = &p_{t_{(i,j,1)}^1}p_{t_{(i,j,2)}^1}(\tau_{t_{(i,j,1)}^1}+\tau_{t_{(i,j,2)}^1}+w_i^c+2t_i^c)+p_{t_{(i,j,1)}^1}(1-p_{t_{(i,j,2)}^1})p_{t_{(i,j,2)}^2}(\tau_{t_{(i,j,1)}^1}+\tau_{t_{(i,j,2)}^1}+\tau_{t_{(i,j,2)}^2}+2t_i^c)\\\nonumber
    &+p_{t_{(i,j,1)}^1}(1-p_{t_{(i,j,2)}^1})(1-p_{t_{(i,j,2)}^2})(\tau_{t_{(i,j,1)}^1}+\tau_{t_{(i,j,2)}^1}+\tau_{t_{(i,j,2)}^2})+
    (1-p_{t_{(i,j,1)}^1})\tau_{t_{(i,j,1)}^1} +t_{ij},
\end{align}
where the first four terms capture the average duration from when a seeker is matched with an idle courier until the courier picks up all orders corresponding to cases (i)-(iv) in Figure \ref{fig:time_taker},  and $t_{ij}$ captures the time from when the taker collects all orders in $i$ to dropping off the seeker order in $j$. This completes the proof. 
\end{proof}

The results of Propositions \ref{proposition2_case1_scenario1}-\ref{prop_idle} can be used to derive the expected delivery time. Let $w_{ij}^g$ denote the expected delivery time for an order delivered from origin $i$ to destination by ground delivery services, including the delivery from restaurant $i$ to customer location $j$, from restaurant $i$ to launchpad $j$, and from kiosk $i$ to customer location $j$, then we have
\begin{align} \label{ground_delivery_time}
    w_{ij}^g = p_{s_{(i,j)}}t_{s_{(i,j)}}^{taker} + (1-p_{s_{(i,j)}})t_{s_{(i,j)}}^{idle}.
\end{align}
Since the time required for a restaurant to prepare a meal depends on the supply characteristics and the total demand received by the restaurant, which are exogenously decided, we do not include it in the calculation of customer delivery time. Furthermore, we assume that the matching time is negligible compared to the pickup time, an assumption that is widely used in existing literature, as shown in studies such as \cite{ke2020pricing, castillo2017surge, ke2023supply}. Thus, the expected delivery time for an order with origin $i$ and destination $j$, when delivered by ground delivery services, has been derived as $w_{ij}^g$.

Let $w_{ilkj}^a$ denote the expected delivery time for an order with restaurant $i$ to customer location $j$ allocated to air delivery services using launchpad $l$ and kiosk $k$, then we have
\begin{align}\label{air_delivery_time_ilkj}
    w_{ilkj}^a = w_{il}^{g} + w_{l}^{a}+t_{lk}^a+ w_{kj}^{g},
\end{align}
where $ w_{il}^{g}$ and $ w_{kj}^{g}$ is the average delivery time from the origin $i$ to launchpad $l$ and from kiosk $k$ to destination $j$ by ground couriers respectively, $w_l^a$ is average queuing time at the launchpad derived (to be derived through the queuing model in the next section), and $t_{lk}^a$ is the average time required for a drone to travel from launchpad $l$ to kiosk $k$ based on the direct distance between the two nodes, which can be treated as an exogenous parameter.

Next, we analyze the average shared time for food-delivery orders. Shared time for each seeker refers to periods when the courier is simultaneously picking up or transporting this seeker along with other orders, leading to multiple orders sharing the same courier with this seeker for a designated duration. Understanding shared time is crucial as it helps quantify the actual distances traveled by couriers, which, in turn, influences the operational costs of the delivery platform. It is important to note that, under a courier capacity of three, in some cases each seeker will share with one another order, but in some other cases the seeker will share with another two orders simulations. It is important to distinguish these two cases, especially for the purpose of calculating the travel distance and operational cost. Let $t_{ij}^{s1}$/$t_{ij}^{s2}$ denote the average time for an order delivered from node $i$ to node $j$ shared with one additional order and two additional orders, respectively. We have the following proposition:
\begin{proposition}
    The average shared time for each order delivered from node $i$ to node $j$ shared with one additional order and two additional orders, respectively, can be characterized as:
\begin{align}\nonumber
    t_{ij}^{s1} = &\rho_{t_{(i,j,1)}^1}\Big[p_{t_{(i,j,2)}^1}(\tau_{t_{(i,j,2)}^1}+t_j^c)+(1-p_{t_{(i,j,2)}^1})p_{t_{(i,j,2)}^2}(\tau_{t_{(i,j,2)}^1}+\tau_{t_{(i,j,2)}^2}+t_j^c)+\\ \nonumber
    &(1-p_{t_{(i,j,2)}^1})(1-p_{t_{(i,j,2)}^2})(\tau_{t_{(i,j,2)}^1}+\tau_{t_{(i,j,2)}^2}+t_{ij})\Big] + \rho_{t_{(i,j,2)}^1}t_j^c+\rho_{t_{(i,j,2)}^2}t_j^c + \\ \nonumber
    & (1-p_{s_{(i,j)}})\Big[p_{t_{(i,j,1)}^1}p_{t_{(i,j,2)}^1}\tau_{t_{(i,j,2)}^1}+p_{t_{(i,j,1)}^1}(1-p_{t_{(i,j,2)}^1})p_{t_{(i,j,2)}^2}(\tau_{t_{(i,j,2)}^1}+\tau_{t_{(i,j,2)}^2})\\
    &+p_{t_{(i,j,1)}^1}(1-p_{t_{(i,j,2)}^1})(1-p_{t_{(i,j,2)}^2})(\tau_{t_{(i,j,2)}^1}+\tau_{t_{(i,j,2)}^2}+t_{ij})\Big] \label{eq:t_s1}\\ \nonumber
    t_{ij}^{s2} = &\rho_{t_{(i,j,1)}^1}\Big[p_{t_{(i,j,2)}^1}(w_i^c+2t_i^c+t_{ij})+(1-p_{t_{(i,j,2)}^1})p_{t_{(i,j,2)}^2}(2t_i^c+t_{ij})\Big] + \rho_{t_{(i,j,2)}^1}(w_i^c+2t_i^c+t_{ij})+\\ 
    &\rho_{t_{(i,j,2)}^2}(2t_i^c+t_{ij}) +  (1-p_{s_{(i,j)}})\Big[p_{t_{(i,j,1)}^1}p_{t_{(i,j,2)}^1}(w_i^c+2t_i^c+t_{ij})+p_{t_{(i,j,1)}^1}(1-p_{t_{(i,j,2)}^1})p_{t_{(i,j,2)}^2}(2t_i^c+t_{ij})\Big].\label{eq:t_s2}
\end{align}
\end{proposition}

\begin{proof}
Similar to the analysis of expected delivery time, we consider two distinct cases. The first cases is that the order is matched to a taker when it is at a seeker state $s_{(i,j)}$. In this case, there are three further scenarios, including (a) the order is matched to a taker in state $t_{(i,j,1)}^1$; (b) the order is matched to a taker in state $t_{(i,j,2)}^1$; and (c) the order is matched to a taker in state $t_{(i,j,2)}^2$. If no takers in these states are available, the order will then be matched to an idle courier and the courier turns into a taker, which represents the second case. Below we will characterize the shared time for each case.  

First, when a seeker in state $s_{(i,j)}$ is matched with a taker in state $t_{(i,j,1)}^1$, the taker will then turn into state $t_{(i,j,2)}^1$, and the movement of the taker has been illustrated in Figure \ref{fig:time_seeker}, with the following three scenarios.  The shared time for each scenario can be delineated as follows:
\begin{enumerate}
    \item $(i,j,2,1)\rightarrow (i,j,3,1) \rightarrow (i,j,3,2) \rightarrow (i,j,3,3)\rightarrow \text{deliver}$: When the courier is at state $(i,j,2,1)$ or en route to deliver the seeker after dropping off the first order, the seeker shares the courier with just one order.
    Thus, the average time shared by the seeker and a single order are $\tau_{t_{(i,j,2)}^1}+t_j^c$.
    When the courier is in states $(i,j,3,1)$, $(i,j,3,2)$, $(i,j,3,3)$, or on its way to deliver the first order, the seeker shares the courier with two other orders.
    Therefore, the average time shared by the seeker and two orders are $w_i^c+2t_i^c+t_{ij}$.
    \item $(i,j,2,1)\rightarrow (i,j,2,2) \rightarrow (i,j,3,2) \rightarrow (i,j,3,3)\rightarrow \text{deliver}$: The seeker shares the courier with only one order when the courier is in state $(i,j,2,1)$, $(i,j,2,2)$, or on the way to deliver the seeker after dropping off the first order.
    The average time shared by the seeker and one order is $\tau_{t_{(i,j,2)}^1}+\tau_{t_{(i,j,2)}^2}+t_j^c$.
    When the courier is in states $(i,j,3,2)$, $(i,j,3,3)$, or on its way to deliver the first order, the seeker shares the courier with two other orders. Consequently, the average shared time between the seeker and two orders is $2t_i^c+t_{ij}$.
    \item $(i,j,2,1)\rightarrow (i,j,2,2) \rightarrow \text{deliver}$:   Given a bundle size of two, the seeker can only share the courier with a single order. This sharing occurs when the courier is at state $(i,j,2,1)$, $(i,j,2,2)$, or en route to deliver the first order. Therefore, the average shared time between the seeker and a single order is $\tau_{t_{(i,j,2)}^1}+\tau_{t_{(i,j,2)}^2}+t_{ij}$, while the average shared time between the seeker and two orders is 0.
\end{enumerate}
When the seeker is matched to taker $t_{(i,j,2)}^1$, it is the last to be picked up and the last to be dropped off. 
The average time shared by the seeker and one order is $t_j^c$, which occurs only when the courier is dropping off the second order after delivering the first.
When the courier is in state $(i,j,3,1)$, $(i,j,3,2)$, $(i,j,3,3)$, or on the way to deliver the first order, it is shared by the seeker and two other orders.
Thus the average time shared by the seeker and two orders is $w_i^c+2t_i^c+t_{ij}$.

When the seeker is matched to taker $t_{(i,j,2)}^2$, similar to the matching with taker $t_{(i,j,2)}^1$, the average time shared by the seeker and one order is $t_j^c$.
The seeker share the courier with other two orders when the courier is in state $(i,j,3,2)$, $(i,j,3,3)$ or en route to drop off the first order,
then the average time shared by the seeker and two orders is $2t_i^c+t_{ij}$.

Finally, we consider the case when the seeker in state $s_{(i,j)}$ is dispatched to an idle courier and turns into a taker state $t_{(i,j,1)}^1$, as shown in Figure \ref{fig:time_taker}.  In this case, the seeker is the first to be both picked up and dropped off. The time spent by the taker with one bundled order and two bundled orders can be  delineated as follows:
\begin{enumerate}
    \item $(i,j,1,1)\rightarrow (i,j,2,1) \rightarrow (i,j,3,1) \rightarrow (i,j,3,2)\rightarrow (i,j,3,3)\rightarrow \text{deliver}$:  As the seeker share the delivery time with a single order only when the courier is in state $(i,j,2,1)$,
    the average time shared by the seeker and one order is $\tau_{t_{(i,j,2)}^1}$.
    If the courier is at state $(i,j,3,1)$, $(i,j,3,2)$, $(i,j,3,3)$ or on the way to deliver the seeker, the seeker share the time with two other orders.
    The average time shared by the seeker and two bundled orders is $w_i^c+2t_i^c+t_{ij}$.
    
    \item $(i,j,1,1)\rightarrow (i,j,2,1) \rightarrow (i,j,2,2) \rightarrow (i,j,3,2)\rightarrow (i,j,3,3)\rightarrow \text{deliver}$:  When the courier is at state $(i,j,2,1)$ or $(i,j,2,2)$, it is shared by the seeker and another order, thus the average time shared by the seeker and one order is $\tau_{t_{(i,j,2)}^1}+\tau_{t_{(i,j,2)}^2}$.
    When the courier is at state $(i,j,3,2)$, $(i,j,3,3)$ or en route to deliver the seeker after collecting all orders, it is shared by the seeker and other two orders.
    The average time shared by the seeker and two orders is $2t_i^c+t_{ij}$.
    
    \item $(i,j,1,1)\rightarrow (i,j,2,1) \rightarrow (i,j,2,2) \rightarrow \text{deliver}$: Given that only two orders are dispatched to the couriers before starting delivery,  the seeker can only share the delivery time with another single order.
    The average time shared by the seeker and one order is $\tau_{t_{(i,j,2)}^1}+\tau_{t_{(i,j,2)}^2}$, and the average time shared by the seeker and two orders is $0$.
    \item $(i,j,1,1) \rightarrow \text{deliver}$: In this case, the courier does not receive any other orders before picking up the seeker, resulting in an average time shared by the seeker and other orders of $0$.  
\end{enumerate}
Based on the aforementioned discussion, we can combine all the occasions where the seeker have shared time with a single order to derive $t_{ij}^{s1}$, and combine all the occasions where the seeker have shared time with two other orders simultaneously to derive $t_{ij}^{s2}$, as follows:
\begin{align}\nonumber
    t_{ij}^{s1} = &\rho_{t_{(i,j,1)}^1}\Big[p_{t_{(i,j,2)}^1}(\tau_{t_{(i,j,2)}^1}+t_j^c)+(1-p_{t_{(i,j,2)}^1})p_{t_{(i,j,2)}^2}(\tau_{t_{(i,j,2)}^1}+\tau_{t_{(i,j,2)}^2}+t_j^c)+\\ \nonumber
    &(1-p_{t_{(i,j,2)}^1})(1-p_{t_{(i,j,2)}^2})(\tau_{t_{(i,j,2)}^1}+\tau_{t_{(i,j,2)}^2}+t_{ij})\Big] + \rho_{t_{(i,j,2)}^1}t_j^c+\rho_{t_{(i,j,2)}^2}t_j^c + \\ \nonumber
    & (1-p_{s_{(i,j)}})\Big[p_{t_{(i,j,1)}^1}p_{t_{(i,j,2)}^1}\tau_{t_{(i,j,2)}^1}+p_{t_{(i,j,1)}^1}(1-p_{t_{(i,j,2)}^1})p_{t_{(i,j,2)}^2}(\tau_{t_{(i,j,2)}^1}+\tau_{t_{(i,j,2)}^2})\\
    &+p_{t_{(i,j,1)}^1}(1-p_{t_{(i,j,2)}^1})(1-p_{t_{(i,j,2)}^2})(\tau_{t_{(i,j,2)}^1}+\tau_{t_{(i,j,2)}^2}+t_{ij})\Big]\label{eq:t_s1}\\ \nonumber
    t_{ij}^{s2} = &\rho_{t_{(i,j,1)}^1}\Big[p_{t_{(i,j,2)}^1}(w_i^c+2t_i^c+t_{ij})+(1-p_{t_{(i,j,2)}^1})p_{t_{(i,j,2)}^2}(2t_i^c+t_{ij})\Big] + \rho_{t_{(i,j,2)}^1}(w_i^c+2t_i^c+t_{ij})+\\ 
    &\rho_{t_{(i,j,2)}^2}(2t_i^c+t_{ij}) +  (1-p_{s_{(i,j)}})\Big[p_{t_{(i,j,1)}^1}p_{t_{(i,j,2)}^1}(w_i^c+2t_i^c+t_{ij})+p_{t_{(i,j,1)}^1}(1-p_{t_{(i,j,2)}^1})p_{t_{(i,j,2)}^2}(2t_i^c+t_{ij})\Big],\label{eq:t_s2}
\end{align}
where the first three terms in (\ref{eq:t_s1}) and (\ref{eq:t_s2}) capture the average shared time for a seeker matched to a taker in state $t_{(i,j,1)}^1$, $t_{(i,j,2)}^1$ and $t_{(i,j,2)}^2$ with probabilities $\rho_{t_{(i,j,1)}^1}$, $\rho_{t_{(i,j,2)}^1}$ and $\rho_{t_{(i,j,2)}^2}$, while the last term in  (\ref{eq:t_s1}) and (\ref{eq:t_s2}) capture the average shared time for a seeker matched to an idle courier with probability $1-p_{s_{(i,j)}}$. This completes the proof. 
\end{proof}

\subsection{Matching between Orders and Drones}

In this subsection, we evaluate the average waiting time for air delivery orders at each launchpad which depends on the matching process between orders and drones.
We employ a double-ended queue \cite{srivastava1982special} to characterize the matching between food-delivery orders and drones at the launchpad. 
In the double-ended queue, if there is no idle drone at the launchpad,  incoming food-delivery orders have to wait in the queue for an idle drone to provide services. Conversely, if there is no food-delivery orders at the launchpad, the idle drones will queue up at the launchpad awaiting new food-delivery orders.
Once the food-delivery orders and idle drones exists at the launchpad simultaneously, the drone will be loaded with the order and takes off for delivery.
The queuing process at the launchpad is illustrated in Figure \ref{fig:queue}.
For simplicity and stability of the double-ended queue,
we have the following assumption:
\begin{assumption}\label{assump:launchpad_queue}
We make the following assumptions regarding the queuing process at each launchpad:
\begin{enumerate}
    \item The capacity of food-delivery orders waiting at the launchpad is limited to $M$. Once this capacity is reached, any additional incoming orders must be serviced by ground delivery.

    \item The arrival rate of food-delivery orders at each launchpad exceeds the arrival rate of drones, specifically $\nu_l^o > \nu_l^d$, ensuring that the average drone queue length remains finite.

    \item The loading time of the orders is negligible compared with the flight time.
\end{enumerate}
\end{assumption}

\begin{figure}[htb!]
    \centering
    \includegraphics[width=0.7\linewidth]{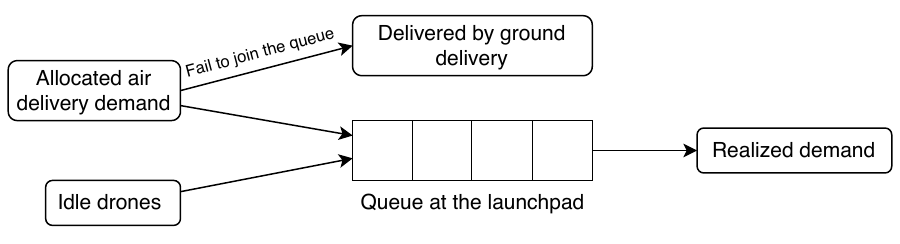}
    \caption{The double-ended queuing process at the launchpad.}
    \label{fig:queue}
\end{figure}

Assumption \ref{assump:launchpad_queue}-(i) prevents the orders from experiencing prolonged waiting times caused by excessively long queues, which is crucial for ensuring the service quality of air delivery services.
Assumptions \ref{assump:launchpad_queue}-(ii) is a mild constraint to guarantee the stability of the double-ended queue, while Assumption \ref{assump:launchpad_queue}-(iii) is also a reasonable assumption since the loading time is short and exogenous and has negligible effects on the entire food-delivery process.  

Assume the arrival of food-delivery orders at the launchpad $l$ is Poisson with rate $\nu_l^o$, and the arrival rate of drones at the launchpad $l$ is Poisson with rate $\nu_l^d$. Let $\tilde \lambda_{kl}^a$ denote the drone repositioning flow from kiosk $k$  to launchpad $l$, and let $\hat\lambda_{ilkj}^a$ denote the intended allocation of food-delivery orders using drone delivery services determined by the platform.
We have
\begin{align} \label{arrival_rate_orders_l}
    &\nu_l^o = \sum_{i\in\mathcal{R}}\sum_{k\in\mathcal{K}}\sum_{j\in\mathcal{C}} \hat \lambda_{ilkj}^a,\\ &\label{arrival_rate_drone_l}
    \nu_l^d = \sum_{k\in\mathcal{K}} \tilde \lambda_{kl}^a,\\
    &\sum_{k\in\mathcal{K}} \tilde \lambda_{kl}^a = \sum_{i\in\mathcal{R}}\sum_{k\in\mathcal{K}}\sum_{j\in\mathcal{C}}\lambda_{ilkj}^a.\label{eq:drone_conservation_l}
\end{align}

\begin{remark}
Note that $\lambda_{ilkj}^a$ is the realized arrival rate of food-delivery orders allocated to drone delivery services with launchpad $l$ and kiosk $k$, while $\hat \lambda_{ilkj}^a\le \lambda_{ij}$ is the intended allocation to drone delivery by the platform without considering the possibility of being rejected at the launchpad due to limited order queue length. It is possible for $\hat \lambda_{ilkj}^a$ to be larger than $\lambda_{ilkj}^a$ since if the queue at a launchpad is excessively long (i.e. reaching $M$), newly assigned orders to that launchpad must be redirected for ground delivery instead.
\end{remark}

Let $n\in [-M,+\infty]\cap \mathbb{R}$ be the length of the queue at the launchpad $l$, where $n < 0$ implies that there are $-n$ food-delivery orders waiting at the launchpad, and $n>0$ implies that there are $n$ drones waiting at the launchpad. Let $P_n^l$ be the probability that the queue length is $n$ at launchpad $l$, and let $\gamma_l = \frac{\nu_l^d}{\nu_l^o} < 1$ be the reliability level of services at the launchpad. Following the properties of the double-ended queue \cite{srivastava1982special}, we have the following relation at the stationary state:
\begin{align}\label{eq:prob_queue_length_l}
    P_n^l = (1-\gamma_l)\gamma_l^{M+n}, \quad\forall l\in\mathcal{L}, \forall n \in [-M,+\infty]\cap \mathbb{R}
\end{align}

Therefore, the average number of drones/orders waiting at launchpad $l$ can be derived as:
\begin{align}
    \bar N_l^d &= \sum_{n=1}^{\infty} n(1-\gamma_l)\gamma_l^{M+n} = \frac{\gamma_l^{M+1}}{1-\gamma_l}, \label{eq:num_drones_l} \\
    \bar N_l^o &= \sum_{n=1}^{M} n(1-\gamma_l)\gamma_l^{M-n} = M-\frac{\gamma_l(1-\gamma_l^M)}{1-\gamma_l}.\label{eq:num_orders_l}
\end{align}

Based on (\ref{eq:prob_queue_length_l}), the probability for the launchpad to have $M$ food-delivery orders waiting at stationary state is:
\begin{align}\label{eq:prob_M_l}
    P_{-M}^l = 1-\gamma_l.
\end{align}

Thus the arrival rate of orders with origin $i$ and destination $j$ that successfully join the queue at launchpad $l$ and are delivered by air delivery services to kiosk $k$ is:
\begin{align}\label{eq:realized_demand_flow}
    \lambda_{ilkj}^{a} = \hat \lambda_{ilkj}^a(1-P_{-M}^l) = \hat \lambda_{ilkj}^{a} \gamma_l = \hat \lambda_{ilkj}^a\frac{\nu_l^d}{\nu_l^o}
\end{align}

By Little's law, the average time from an order being delivered to launchpad $i$ to it being sent out from the launchpad on a drone, denoted by $w_l^a$, is
\begin{align} \label{waiting_time_l}
    w_l^a = \frac{\bar N_l^o}{\sum_{i\in\mathcal{R}}\sum_{k\in\mathcal{L}}\sum_{j\in\mathcal{C}}\lambda_{ilkj}^a} 
\end{align}

Note that (\ref{eq:realized_demand_flow}) is a bilinear constraint, which is nonconvex and challenging to address in an optimization problem.
To linearize (\ref{eq:realized_demand_flow}), we discretize $\gamma_l$ into several reliability levels $(\gamma_r)_{r\in\mathcal{RL}}$, e.g., $\gamma_1 = 0,\ \gamma_2 = 0.1,\dots,\gamma_{10} = 0.9$, $\mathcal{RL} = \{1,2,\dots,10\}$. Let the binary variable $\omega_{lr}$ indicate which reliability level the launchpad $l$ is in, where
\begin{align}
    \omega_{lr} = \begin{cases}
        1, & \gamma_l = \gamma_r\\
        0, & \text{otherwise}
    \end{cases}.
\end{align}
Thus, equations (\ref{eq:num_drones_l})-(\ref{eq:realized_demand_flow}) can be approximated by the following linear equations:
\begin{align} \label{eq:index_rl_sum_1}
    &\sum_{r\in\mathcal{RL}} \omega_{lr} = 1,\\  \label{eq:gamma_l_obtain}
    &\bar N_l^d  = \sum_{r\in\mathcal{RL}}\omega_{lr}\frac{\gamma_r^{M+1}}{1-\gamma_r}, \\
    &\bar N_l^o =  M-\sum_{r\in\mathcal{RL}}\frac{\omega_{lr}\gamma_r(1-\gamma_r^M)}{1-\gamma_l}\\
    \label{eq:lb_big_M_gamma_linearze}
    & \lambda_{ilkj}^a \ge \gamma_r\hat \lambda_{ilkj}^a - (1-\omega_{lr})M_{\gamma},\quad\forall r\in\mathcal{RL}\\
    &\lambda_{ilkj}^a\le \gamma_r\hat \lambda_{ilkj}^a + (1-\omega_{lr})M_{\gamma},\quad\forall r\in\mathcal{RL} \label{eq:ub_big_M_gamma_linearze} 
\end{align}
where $M_{\gamma}$ is a sufficiently large number.

\subsection{Couriers Fleet Size}
Each courier may operate in one of the three modes: (1) dispatched couriers either on the way to pick up or deliver the food-delivery orders, (2) idle couriers cruising within some zone for new food-delivery orders, and (3) idle couriers repositioning to other zones for new food-deliver orders.
For simplicity, we assume that driver repositioning decisions are managed by the food-delivery platform. This assumption is rational as the platform can efficiently communicate incentive details to idle couriers, motivating them to relocate to zones with better earning prospects. 
Recall that in section \ref{sec:bundle_prob}, we have divided the transportation network into several zones denoted by the set $\mathcal{Z}$, with $Z(i)$ representing the corresponding zone that node $i$ is in.
Let $q_{Z_1Z_2}^{I,r}$ denote the flow of idle drivers repositioning from zone $Z_1$ to $Z_2$ per unit of time, then in the steady-state, we have the following courier flow conservation constraint for each zone $Z_1$:
\begin{align}
    \sum_{Z_2\in\mathcal{Z}}q_{Z_1Z_2}^{I,r} + \sum_{i\in\{v:Z(v)=Z_1\}}\sum_{j\in\mathcal{L}\cup\mathcal{C}}\lambda_{ij}^{g,all} = \sum_{Z_2\in\mathcal{Z}}q_{Z_2Z_1}^{I,r} + \sum_{i\in\mathcal{R}\cup\mathcal{K}}\sum_{j\in\{v:Z(v)=Z_1\}}\lambda_{ij}^{g,all},
\end{align}
where the left-hand side is the total outflow that originates from zone $Z_1$, and the right-hand side is the total inflow that are destined to zone $Z_1$. 
Let $t_{Z_1Z_2}^r$ be the average repositioning time for idle couriers from zone $Z_1$ to zone $Z_2$, and let $N^r$ be the average number of repositioning couriers. Based on Little's Law,  we have:
\begin{align}
    N^r = \sum_{Z_1\in\mathcal{Z}}\sum_{Z_2\in\mathcal{Z}}q_{Z_1Z_2}^{I,r}t_{Z_1Z_2}^r.
\end{align}

Let $N^c$ be the average number of occupied couriers that are picking up or delivering food-delivery orders, which equals to the total delivery time of food-delivery orders delivered by couriers from each OD pair minus the shared time caused by order bundling which was counted twice or three times in the total delivery time. {Based on Little's law, the total delivery time for food-delivery orders can be derived as $\sum_{i\in\mathcal{R}\cup\mathcal{K}}\sum_{j\in\mathcal{L}\cup\mathcal{C}}\lambda_{ij}^{g,all}w_{ij}^g$. However, since two or more orders may share the same courier at the same time, this term is larger than the actual number of occupied courier on the platform. To derive the number of couriers, these repetition due to shared time should be deducted from to $\sum_{i\in\mathcal{R}\cup\mathcal{K}}\sum_{j\in\mathcal{L}\cup\mathcal{C}}\lambda_{ij}^{g,all}w_{ij}^g$. In particular, we have the following relation:
\begin{align}
\label{courier_number}
    N^c  =\sum_{i\in\mathcal{R}\cup\mathcal{K}}\sum_{j\in\mathcal{L}\cup\mathcal{C}}\lambda_{ij}^{g,all}w_{ij}^g-\sum_{i\in\mathcal{R}\cup\mathcal{K}}\sum_{j\in\mathcal{L}\cup\mathcal{C}}\left(\frac{1}{2} \lambda_{ij}^{g,all}t_{ij}^{s1}+\frac{2}{3} \lambda_{ij}^{g,all}t_{ij}^{s_2}\right).
\end{align}
where the term 
$\sum_{i \in \mathcal{R} \cup \mathcal{K}} \sum_{j \in \mathcal{L} \cup \mathcal{C}} \lambda_{ij}^{g,\text{all}} t_{ij}^{s1}$
is the summation of shared time over all orders that have shared with another (single) order. Since the shared time is included for each of the two orders, it is actually twice the corresponding courier serving time. Similarly, the term 
$\sum_{i \in \mathcal{R} \cup \mathcal{K}} \sum_{j \in \mathcal{L} \cup \mathcal{C}} \lambda_{ij}^{g,\text{all}} t_{ij}^{s2}$ 
is the summation of shared time over all orders that have shared with two other orders simultaneously, which is three times the corresponding courier serving time. Therefore, the second term on the right-hand side of Equation~(\ref{courier_number}), 
$
\sum_{i \in \mathcal{R} \cup \mathcal{K}} \sum_{j \in \mathcal{L} \cup \mathcal{C}} \left( \frac{1}{2} \lambda_{ij}^{g,\text{all}} t_{ij}^{s1} + \frac{2}{3} \lambda_{ij}^{g,\text{all}} t_{ij}^{s2} \right),
$ 
is exactly the shared time that has been over-counted. }

Based on the above discussion, the total number of couriers should satisfy:
\begin{align} \nonumber
     N &= \sum_{Z\in\mathcal{Z}}N_Z^I + N^c + N^r \\
     &= \sum_{Z\in\mathcal{Z}} N_Z^I+\sum_{i\in\mathcal{R}\cup\mathcal{K}}\sum_{j\in\mathcal{L}\cup\mathcal{C}}\lambda_{ij}^{g,all}w_{ij}^g-\sum_{i\in\mathcal{R}\cup\mathcal{K}}\sum_{j\in\mathcal{L}\cup\mathcal{C}}\left(\frac{1}{2} \lambda_{ij}^{g,all}t_{ij}^{s1}+\frac{2}{3} \lambda_{ij}^{g,all}t_{ij}^{s_2}\right) + \sum_{Z_1\in\mathcal{Z}}\sum_{Z_2\in\mathcal{Z}}q_{Z_1Z_2}^{I,r}t_{Z_1Z_2}^r.\label{num_drivers}
\end{align}

\subsection{Drone Fleet Size}
Each drone may also operate in three distinct modes: (1) waiting at the launchpad, (2) delivering food-delivery orders from a launchpad to a kiosk, and (3) relocating to a new launchpad after dropping off the order at a kiosk.
Consequently, the total number of drones operated by the platform (denoted by $N^a$) can be characterized by the following conservation constraint:
\begin{align}  \label{num_drones}
    N^a = \sum_{l\in\mathcal{L}} \bar N_l^{d,I} + \sum_{i\in\mathcal{R}}\sum_{l\in\mathcal{L}}\sum_{k\in\mathcal{K}}\sum_{j\in\mathcal{C}} \lambda_{ilkj}^at_{lk}^a +\sum_{k\in\mathcal{K}\cup \mathcal{L}}\sum_{l\in\mathcal{L}}t_{kl}^a\tilde\lambda_{kl}^a.
\end{align}
On the ride-hand side of (\ref{num_drones}), the first term is the number of idle drones at launchpads, the second term is the number of drones carrying a package and flying to the kiosk,  and the last term is the number of drones relocating to a new launchpad after it drops off the package at the kiosk.
Since the drones can not stop at the kiosk but have to return to a launchpad after it finishes the delivery task, we have the following flow-conservation constraints at kiosks:
\begin{align}
    \sum_{i\in \mathcal{R}}\sum_{l\in \mathcal{L}}\sum_{j\in\mathcal{C}}\lambda_{ilkj}^a = \sum_{l\in\mathcal{L}} \tilde \lambda_{kl}^a, \quad \forall k\in\mathcal{K}
\end{align}
which dictates that the flows that arrives at each kiosk must be equal to the repositioning flows out of this kiosk. 
In addition, due to the limitation of the battery capacity, we restrict the maximum allowable flight time for drone delivery and relocation, denoted by $\bar t_d$ and $\bar t_r$, by the following constraints:
\begin{align}
    &\lambda_{ilkj}^a(t_{lk}^a-\bar t_d)\le 0,\quad \forall i\in\mathcal{R}, j\in\mathcal{C},l\in\mathcal{L},k\in\mathcal{K}\\
    &\tilde \lambda_{kl}^a(t_{lk}^a-\bar t_r)\le 0, \quad \forall l\in\mathcal{L}, k\in\mathcal{K}
\end{align}

\subsection{Platform Decisions}
The food-delivery platform determines the location of launchpads $y_l$ ($l\in\mathcal{L})$, the location of kiosks $z_k$ ($k\in\mathcal{C}$),
the allocation of demand to air delivery mode $\hat \lambda_{ilkj}^a$, the repositioning of drones from kiosks to launchpads $\tilde \lambda_{kl}^a$,  the repositioning flow of idle couriers from a zone to another ($q_{Z_1Z_2}^{I,r}$) and the number of idle drivers in each zone $N_Z^I$, to minimize an objective function that represents a trade-off of the operation costs and delivery time. 
Let $C_L^l$/$C_K^k$ be the average construction costs for a launchpad/kiosk per unit of time,  and $C^a$ be the average investment and operation costs of drones per unit of time.
Let $q$ be the average wage paid to couriers per unit of time.
The profit maximization problem of the food-delivery platform can be then formulated as follows:
\newpage
\begin{align} 
    \min_{\bm{y},\bm{z},\bm{\hat \lambda}^a,\bm{\tilde\lambda}^a,\bm{q}^{I,r},\bm{N}^I}\quad\Pi = &\sum_{l\in \mathcal{L}} C_L^ly_l + \sum_{k\in \mathcal{C}} C_K^kz_k + 
    qN+C^aN^a +\alpha_w\frac{\sum_{i\in\mathcal{R}}\sum_{l\in\mathcal{L}}\sum_{k\in\mathcal{K}}\sum_{j\in\mathcal{C}}\lambda_{ilkj}^aw_{ilkj}^a+\lambda_{ij}^gw_{ij}^g}{\sum_{i\in\mathcal{R}}\sum_{j\in\mathcal{C}}\lambda_{ij}} \label{objective_function}
\end{align}
\begin{subnumcases}{\label{constraint}}
    \lambda_{ij} = \sum_{l\in\mathcal{L}}\sum_{k\in\mathcal{K}} \lambda_{ilkj}^a+\lambda_{ij}^g,\quad\forall i\in\mathcal{R}, j\in\mathcal{C} \label{constraint_demand_allocation}\\
    \sum_{i\in\mathcal{R}}\sum_{j\in\mathcal{C}}\sum_{k\in\mathcal{K}}\hat\lambda^a_{ilkj}\le y_lM_L^l,\quad \forall l\in\mathcal{L} \label{constraint_cap_allocation_l} \\
    \sum_{i\in\mathcal{R}}\sum_{l\in\mathcal{L}}\sum_{j\in\mathcal{C}}\hat \lambda^a_{ilkj}\le z_kM_K^k,\quad \forall  k\in\mathcal{K}\label{constraint_cap_allocation_k} \\
     \lambda_{ij}^{g,all} = 
        \lambda_{ij}^g+b_1\sum_{k\in\mathcal{K}} \sum_{j'\in\mathcal{C}}\lambda_{ijkj'}^a+b_2\sum_{i'\in\mathcal{R}}\sum_{l\in\mathcal{L}}\lambda_{i'lij}^a \label{constraint_all_ground_delivery}\\
     \sum_{i\in \mathcal{R}}\sum_{l\in \mathcal{L}}\sum_{j\in\mathcal{C}}\lambda_{ilkj}^a = \sum_{l\in\mathcal{L}} \tilde \lambda_{kl}^a, \quad \forall k\in\mathcal{K} \label{constraint_flow_conservation_k} \\
     \sum_{k\in\mathcal{K}} \tilde \lambda_{kl}^a = \sum_{i\in\mathcal{R}}\sum_{k\in\mathcal{K}}\sum_{j\in\mathcal{C}}\lambda_{ilkj}^a ,\quad\forall l\in\mathcal{L}\label{constraint_flow_conservation_l}\\
     \sum_{k\in\mathcal{K}}\tilde\lambda_{kl}^a = \nu_l^d,\quad \forall l\in\mathcal{L} \label{constraint_order_flow_conservation_l_drone}\\
     \sum_{i\in\mathcal{R}}\sum_{k\in\mathcal{K}}\sum_{j\in\mathcal{C}}\hat\lambda_{ilkj}^a = \nu_l^0,\quad\forall l\in\mathcal{L}\label{constraint_order_flow_conservation_l_order}\\
      \label{constraint_sum_omega_1_linearize}
    \sum_{r\in\mathcal{RL}} \omega_{lr} = 1,\quad\forall l\in\mathcal{L}\\ 
    \label{constraint_lb_big_M_gamma_linearze}
     \lambda_{ilkj}^a \ge \gamma_r\hat \lambda_{ilkj}^a - (1-\omega_{lr})M_{\gamma},\quad\forall r\in\mathcal{RL}\\
    \lambda_{ilkj}^a\le \gamma_r\hat \lambda_{ilkj}^a + (1-\omega_{lr})M_{\gamma},\quad\forall r\in\mathcal{RL} \label{constraint_ub_big_M_gamma_linearze}\\
     \bar N_l^d = \sum_{r\in\mathcal{RL}}\frac{\gamma_r^{M+1}}{1-\gamma_r}\omega_{lr},\quad\forall l\in\mathcal{L} \label{constraint_num_drone_l}\\
     \bar N_l^o = M-\sum_{r\in\mathcal{RL}}\omega_{lr}\frac{\gamma_r(1-\gamma_r^M)}{1-\gamma_r} \label{constraint_num_order_l}\\
      w_l^a = \frac{ \bar N_l^o}{\sum_{i\in\mathcal{R}}\sum_{k\in\mathcal{K}}\sum_{j\in\mathcal{C}}\lambda_{ilkj}^a},\quad\forall l\in\mathcal{L} \label{constraint_wait_time_l}\\
    \lambda_{ilkj}^a(t_{lk}^a-\bar t_d)\le 0,\quad \forall i\in\mathcal{R}, j\in\mathcal{C},l\in\mathcal{L},  k\in\mathcal{K}\label{constraint_battery_delivery}\\
    \tilde \lambda_{kl}^a(t_{lk}^a-\bar t_r)\le 0, \quad \forall l\in\mathcal{L}, k\in\mathcal{K}\cup\mathcal{L}    \label{constraint_battery_return}  \\
      (\ref{eq:pickup_time})-(\ref{eq:t_s2}) \label{constraint_delivery_share_time} \\
      N =  \sum_{Z\in\mathcal{Z}} N_Z^I+\sum_{i}\sum_{j}\lambda_{ij}^{g,all}w_{ij}^g-\sum_{i}\sum_{j}\left(\frac{1}{2} \lambda_{ij}^{g,all}t_{ij}^{s1}+\frac{2}{3} \lambda_{ij}^{g,all}t_{ij}^{s_2}\right) + \sum_{Z_1\in\mathcal{Z}}\sum_{Z_2\in\mathcal{Z}}q_{Z_1Z_2}^{I,r}t_{Z_1Z_2}^r \label{constraint_num_courier} \\
      \sum_{Z_2\in\mathcal{Z}}q_{Z_1Z_2}^{I,r} + \sum_{Z(i)=Z_1}\sum_{j\in\mathcal{L}\cup\mathcal{C}}\lambda_{ij}^{g,all} = \sum_{Z_2\in\mathcal{Z}}q_{Z_2Z_1}^{I,r} + \sum_{i\in\mathcal{R}\cup\mathcal{K}}\sum_{Z(j)=Z_1}\lambda_{ij}^{g,all},\quad \forall Z_1\in\mathcal{Z} \label{constraint_courier_conservation}\\
      N^a = \sum_{l\in\mathcal{L}} \bar N_l^{d,I} + \sum_{i\in\mathcal{R}}\sum_{l\in\mathcal{L}}\sum_{k\in\mathcal{K}}\sum_{j\in\mathcal{C}} \lambda_{ilkj}^at_{lk}^a+\sum_{k\in\mathcal{K}\cup\mathcal{L}}\sum_{l\in\mathcal{L}}t_{kl}^a\tilde\lambda_{kl}^a \label{constraint_num_drones} \\
      y_l,z_k \in \{0,1\},\quad \forall i,\in \mathcal{R},k\in\mathcal{K} \\
      \omega_{lr}\in\{0,1\},\quad\forall l\in\mathcal{L},r\in\mathcal{RL}\\
      \bm{\hat \lambda}^a,\bm{\tilde\lambda}^a,\bm{q}^{I,r},\bm{N}^I\ge 0\label{constraint_positive}
\end{subnumcases}
The objective function (\ref{objective_function}) balances operational costs against service quality, which is measured by the average delivery time. It incorporates launchpad construction costs $C_L^ly_l$, kiosk construction costs $C_K^kz_k$, total wages paid to the couriers $qN$, drone investment and operation costs $C^aN^a$, and the average delivery time multiplied by the trade-off parameter $\alpha_w\frac{\sum_{i\in\mathcal{R}}\sum_{l\in\mathcal{L}}\sum_{k\in\mathcal{K}}\sum_{j\in\mathcal{C}}\lambda_{ilkj}^aw_{ilkj}^a+\lambda_{ij}^gw_{ij}^g}{\sum_{i\in\mathcal{R}}\sum_{j\in\mathcal{C}}\lambda_{ij}}$. Constraints (\ref{constraint_demand_allocation})-(\ref{constraint_cap_allocation_k}) capture the allocation of food-delivery demand to distinct delivery modes, and the assignment of launchpads/kiosks for air delivery services. Constraint (\ref{constraint_all_ground_delivery}) establishes the relationship between overall ground flow and food-delivery demand in different modes.  Constraints (\ref{constraint_flow_conservation_k}) and (\ref{constraint_flow_conservation_l}) enforce the conservation at kiosks and launchpads. Constraints (\ref{constraint_order_flow_conservation_l_drone}) and (\ref{constraint_order_flow_conservation_l_order}) characterize the arrival rate of orders and drones at each launchpad. Constraints (\ref{constraint_sum_omega_1_linearize})-(\ref{constraint_ub_big_M_gamma_linearze}) specify the reliable level and the realized air delivery demand at each launchpad. Constraint (\ref{constraint_num_drone_l}) and (\ref{constraint_num_order_l}) capture the number of idle drones and orders waiting at each launchpad, respectively, and constraint (\ref{constraint_wait_time_l}) captures the average waiting time of food-delivery orders at the launchpad. Constraints (\ref{constraint_battery_delivery}) and (\ref{constraint_battery_return}) ensure the flight time is within the battery limitation. Constraint (\ref{constraint_delivery_share_time}) characterizes the average and shared delivery times for ground delivery flows considering order bundling. Constraint (\ref{constraint_num_courier}) ensures that the number of couriers equals the sum of the three courier operating modes, and constraint (\ref{constraint_courier_conservation}) captures the courier flow conservation across various zones. Constraint (\ref{constraint_num_drones}) ensures that the number of drones aligns with the sum of the three drone operating modes. Overall, this is a mixed-integer nonlinear program.

\section{Solution Method}

The cost-minimization problem, as defined in (\ref{objective_function}), poses a significant challenge as it is categorized under MINLP. These problems are notoriously difficult to solve efficiently with state-of-the-art algorithms, often demanding extensive computational resources and considerable time to achieve an acceptable solution. In this section, we introduce a tailored algorithm within a neural-network-assisted linearization framework. This innovative approach effectively transforms the complex MINLP into a more tractable MILP, with only a negligible optimality gap between the original MINLP and the reformulated MILP. Consequently, this MILP can be efficiently solved to global optimality using commercial solvers.

Motivated by \cite{patel2022neur2sp}, we incorporate an l-layer fully-connected neural network (NN) to approximate the intricate nonlinear relationships within constraint (\ref{constraint_delivery_share_time}) induced by the complex order bundling process, and then embed this neural network as the constraints of the cost- minimization problem (\ref{objective_function}). To facilitate the discussion, we generically denote $\bm{x}$ as the input and $\bm{\Theta}$ as the output of a neural network, and use $ReLU(a) = \max\{0,a\}$ as the activation function. The l-layer fully-connected neural network can be then represented as:
\begin{align}
    \bm{h}^{1} &= ReLU(\bm{W}^{0}\bm{x}+\bm{b}^{0}), \\
    \bm{h}^{m+1} &= ReLU(\bm{W}^{m}\bm{h}^{m}+\bm{b}^{m}),\quad m = 1,\dots,l-1, \\
    \bm{\Theta} &= \bm{W}^{l}\bm{h}^{l}+\bm{b}^{l},
\end{align}
where $\bm{h}^{m} \in \mathbb{R}^{d^{m}}$ is the $m$-th hidden layer,  $\bm{W}^{m}$ is the matrix of weights from layer $m$ to layer $m+1$, $\bm{b}^{m}$ is the bias at the $i$-th layer.
Based on the above notation, we summarize how we transform the optimization problem (\ref{objective_function}) into an MILP framework leveraging the neural network through the following proposition:
\begin{proposition}\label{prop:MILP}
    Let $\mathcal{D}$ be the set of all possible ground flow OD pairs $(i,j)$ with $i\in\mathcal{R}\cup\mathcal{K}$ and $j\in\mathcal{C}\cup\mathcal{L}$.
    For each OD pair $(i,j)\in \mathcal{D}$, learn the mapping from $(\lambda_{ij}^{g,all},N_{Z(i)}^I)$ to $(\lambda_{ij}^{g,all}\bar w_{ij}^g,\frac{1}{2} \lambda_{ij}^{g,all}t_{ij}^{s1}+\frac{2}{3} \lambda_{ij}^{g,all}t_{ij}^{s_2})$ through a l-layer fully connected ReLU network $\Phi_{ij}$ with weights set $\bm{W_{ij}}$ and bias set $\bm{b_{ij}}$. For a given layer $m$ in $\Phi_{ij}$, the $v^{th}$ hidden neuron denoted by $h_{ij}^{m,v}$, can be expressed as
    \begin{align}
        h_{ij}^{m,v} = ReLU\left(\sum_{u=1}^{d^{m-1}}w_{ij}^{m-1,uv}h_{ij}^{m-1,u}+b_{ij}^{m-1,v}\right),
    \end{align}
    where $\bm{W}_{ij}^{m-1}\in\bm{W_{ij}}$ is the matrix of weights from layer $m-1$ to layer $m$, $w_{ij}^{m-1,uv}$ is the $v^{th}$ row and $u^{th}$ column of matrix $\bm{W}_{ij}^{m-1}$, $\bm{b_{ij}}^{m-1}\in\bm{b_{ij}}$ is the vector of bias at the layer $m-1$, $b_{ij}^{m-1,v}$ is the $v^{th}$ element of $\bm{b_{ij}}^{m-1}$, and $d^{m-1}$ is the number of neurons in layer $m-1$. 
    Thus, the optimization problem (\ref{objective_function}) under the neural-network estimation is equivalent to the following MILP:

    \begin{align}\nonumber
&\min_{\bm{z},\bm{y},\bm{\hat \lambda},\bm{\tilde \lambda},\bm{q}^{I,r},\bm{\lambda}^{g,all},\bm{N^I},\bm{\omega},\bm{\hat h},\bm{\check h},\bm{\kappa},\bm{\Theta}}\quad \sum_{l\in \mathcal{L}} C_L^ly_l + \sum_{k\in \mathcal{C}} C_K^kz_k+q\left(\sum_{Z\in\mathcal{Z}} N_Z^I+\sum_{i\in\mathcal{R}\cup\mathcal{K}}\sum_{j\in \mathcal{L}\cup\mathcal{C}} \Theta^1_{ij}  - \sum_{i\in \mathcal{R}\cup\mathcal{K}}\sum_{j\in\mathcal{C}\cup\mathcal{L}} \Theta^2_{ij}\right)+ \\
    &+q\sum_{Z_1\in\mathcal{Z}}\sum_{Z_2\in\mathcal{Z}}q_{Z_1Z_2}^{I,r}t_{Z_1Z_2}^r+C^aN^a + \alpha_w\frac{\sum_{l\in\mathcal{L}}\bar N_l^o+\sum_{i\in\mathcal{R}}\sum_{j\in\mathcal{C}}\sum_{l\in\mathcal{L}}\sum_{k\in\mathcal{K}}t_{lk}^a\lambda_{ilkj}^a+\sum_{i\in\mathcal{R}\cup\mathcal{K}}\sum_{j\in\mathcal{C}\cup\mathcal{L}}\Theta^1_{ij}}{\sum_{i\in\mathcal{R}}\sum_{j\in\mathcal{C}}\lambda_{ij}}\label{objective_MILP}
\end{align}
\begin{subnumcases}{}
    \label{constraint_layer_0_MILP}
    \sum_{u=1}^{2} w_{ij}^{0,uv} x_{ij}^u + b_{ij}^{0,v} =  h_{ij}^{1,v}-\check h_{ij}^{1,v}, \quad \forall (i,j)\in\mathcal{D}, v = 1,\dots,d^1,  \\
    \label{constraint_layer_m_MILP}
    \sum_{u=1}^{d^{m-1}} w_{ij}^{m-1,uv} h_{ij}^{m-1,u} + b_{ij}^{m-1,v}  = h_{ij}^{m,v}-\check h_{ij}^{m,v},\quad \forall (i,j)\in\mathcal{D}, m = 2,\dots, l-1,\ v = 1,\dots, d^{m}, \\
    \label{constraint_layer_predict_MILP}
    {[\Theta_{ij}^1,\Theta_{ij}^2]}^T = W_{ij}^{l}\bm{h}_{ij}^{l}+\bm{b}^{l}, \quad \forall (i,j)\in\mathcal{D}\\
     h_{ij}^{m,u} - M_{nn}(1-\kappa_{ij}^{m,u}) \le 0, \quad \forall (i,j)\in\mathcal{D}, m = 1,\dots, l-1,\ u = 1,\dots, d^{m}, \label{constraint_NN_1_integer_1_MILP}\\
    \check h_{ij}^{m,u} - M_{nn} \kappa_{ij}^{m,u} \le 0, \quad \forall (i,j)\in\mathcal{D}, m = 1,\dots, l-1,\ u = 1,\dots, d^{m}, \label{constraint_NN_1_integer_2_MILP}\\
     h_{ij}^{m,u},\check h_{ij}^{m,u} \ge 0, \quad \forall (i,j)\in\mathcal{D}, m = 1,\dots, l-1,\ u = 1,\dots, d^{m}, \label{constraint_nonegative_h_MILP} \\
    \kappa_{ij}^{m,u} \in\{0,1\}, \quad \forall (i,j)\in\mathcal{D}, m = 1,\dots, l-1,\ u = 1,\dots, d^{m},  \label{constraint_binary_kappa_MILP}\\
    {[x_{ij}^1,x_{ij}^2]} = [\lambda_{ij}^{g,all},N_{Z(i)}^I],\quad \forall (i,j)\in\mathcal{D}, \label{constraint_predict_var_MILP}\\
    \label{constraint_remain_MILP}
    (\ref{constraint_demand_allocation})-(\ref{constraint_num_order_l}),\ (\ref{constraint_battery_delivery})-(\ref{constraint_battery_return}),\ (\ref{constraint_courier_conservation})-(\ref{constraint_positive})
\end{subnumcases}
where $M_{nn}$ is a sufficiently large number.
\end{proposition}
\begin{proof}
Note that the only nonlinear constraints in (\ref{constraint})
arise from  (\ref{constraint_wait_time_l}), (\ref{constraint_delivery_share_time}) and (\ref{constraint_num_courier}).
The latter two, i.e., (\ref{constraint_delivery_share_time}) and (\ref{constraint_num_courier}), can be incorporated into the objective function, leading to the following reformulated cost:
\begin{align}\nonumber
    cost = &\sum_{l\in \mathcal{L}} C_L^ly_l + \sum_{k\in \mathcal{C}} C_K^kz_k+qN_Z^I+q\sum_{i}\sum_{j}\lambda_{ij}^{g,all}\bar w_{ij}^g-q\sum_{i}\sum_{j}\left(\frac{1}{2} \lambda_{ij}^{g,all}t_{ij}^{s1}+\frac{2}{3} \lambda_{ij}^{g,all}t_{ij}^{s_2}\right) + \\ \nonumber
    &q\sum_{Z_1\in\mathcal{Z}}\sum_{Z_2\in\mathcal{Z}}q_{Z_1Z_2}^{I,r}t_{Z_1Z_2}^r+ C^aN^a +\\ 
   & \alpha_w\frac{\sum_{i\in\mathcal{R}}\sum_{j\in\mathcal{C}}\sum_{l\in\mathcal{L}}\sum_{k\in\mathcal{K}}(\bar w_{il}^g+w_l^a+t_{lk}^a+\bar w_{kj}^g)\lambda_{ilkj}^a+\sum_{i\in\mathcal{R}}\sum_{j\in\mathcal{C}}\lambda_{ij}^{g} \bar w_{ij}^g}{\sum_{i\in\mathcal{R}}\sum_{j\in\mathcal{C}}\lambda_{ij}w_i^r{\sum_{i\in\mathcal{R}}\sum_{j\in\mathcal{C}}\lambda_{ij}}}, 
   \label{equa_69}
\end{align}
where $\bar w_{ij}^g$, $t_{ij}^{s1}$ and $t_{ij}^{s2}$ can be regarded as functions of $\lambda_{ij}^{g,all}$ and $N_{Z(i)}^I$, and their relations are characterized by  (\ref{constraint_delivery_share_time}). 
Furthermore, given constraints (\ref{constraint_wait_time_l}) and (\ref{constraint_all_ground_delivery}), we have
\begin{align}
    &\sum_{i\in\mathcal{R}}\sum_{j\in\mathcal{C}}\sum_{l\in\mathcal{L}}\sum_{k\in\mathcal{K}} w_l^a\lambda_{ilkj}^a = \sum_{l\in\mathcal{L}} \bar N_l^o,\label{equa_70}\\
    &\sum_{i\in\mathcal{R}}\sum_{j\in\mathcal{C}}\sum_{l\in\mathcal{L}}\sum_{k\in\mathcal{K}}(\bar w_{ij}^g+\bar w_{kj}^g)\lambda_{ilkj}^a+\sum_{i\in\mathcal{R}}\sum_{j\in\mathcal{C}}\lambda_{ij}^{g} \bar w_{ij}^g = \sum_{i\in\mathcal{R}\cup\mathcal{K}}\sum_{j\in\mathcal{C}\cup\mathcal{L}}\lambda_{ij}^{g,all}\bar w_{ij}^g.\label{equa_71}
\end{align}
where (\ref{equa_71})  is obtained by multiplying $\bar{w}_{ij}^g$ to each OD pair in constraint (\ref{constraint_all_ground_delivery}), and then make the summation over all the OD pairs. 
Plugging (\ref{equa_70}) and (\ref{equa_71}) into the last term of objective function (\ref{equa_69}), the cost can be then rewritten as:
\begin{align}
    cost = &\sum_{l\in \mathcal{L}} C_L^ly_l + \sum_{k\in \mathcal{C}} C_K^kz_k+q\left(\sum_{Z\in\mathcal{Z}} N_Z^I+\sum_{i\in\mathcal{R}\cup\mathcal{K}}\sum_{j\in \mathcal{L}\cup\mathcal{C}} \Theta^1_{ij}  - \sum_{i\in \mathcal{R}\cup\mathcal{K}}\sum_{j\in\mathcal{C}\cup\mathcal{L}} \Theta^2_{ij}\right)+C^aN^a+ \\
    &q\sum_{Z_1\in\mathcal{Z}}\sum_{Z_2\in\mathcal{Z}}q_{Z_1Z_2}^{I,r}t_{Z_1Z_2}^r + \alpha_w\frac{\sum_{l\in\mathcal{L}}\bar N_l^o+\sum_{i\in\mathcal{R}}\sum_{j\in\mathcal{C}}\sum_{l\in\mathcal{L}}\sum_{k\in\mathcal{K}}t_{lk}^a\lambda_{ilkj}^a+\sum_{i\in\mathcal{R}\cup\mathcal{K}}\sum_{j\in\mathcal{C}\cup\mathcal{L}}\Theta^1_{ij}}{\sum_{i\in\mathcal{R}}\sum_{j\in\mathcal{C}}\lambda_{ij}},
\end{align}
where $\Theta^1_{ij} = \lambda_{ij}^{g,all}\bar w_{ij}^g$ and $\Theta^2_{ij} = \frac{1}{2} \lambda_{ij}^{g,all}t_{ij}^{s1}+\frac{2}{3} \lambda_{ij}^{g,all}t_{ij}^{s_2}$.

Note that both $\Theta^1_{ij}$ and $\Theta^2_{ij}$ are nonlinear functions of $\lambda_{ij}^{g,all}$ and $N_{Z(i)}^I$ as delineated in (\ref{eq:pickup_time})-(\ref{eq:t_s2}). To address the nonlinearlity,  we will approximate the relations between $[\Theta^1_{ij},\Theta^2_{ij}]$ and $[\lambda_{ij}^{g,all},N_{Z(i)}^I]$ using a neural network, which despite being nonlinear, can be equivalently transformed into mixed integer linear constraints as described in (\ref{constraint_layer_0_MILP})-(\ref{constraint_predict_var_MILP}).
Specially, for each OD pair $(i,j)\in\mathcal{D}$, we use an l-layered neural network with ReLU activation function to estimate the mapping from $[\lambda_{ij}^{g,all},N_{Z(i)}^I]$ to $[\Theta^1_{ij},\Theta^2_{ij}]$. Graphically, it can be represented as Figure \ref{fig:NN}.
\begin{figure}[t!]
     \centering
     \includegraphics[width=0.8\linewidth]{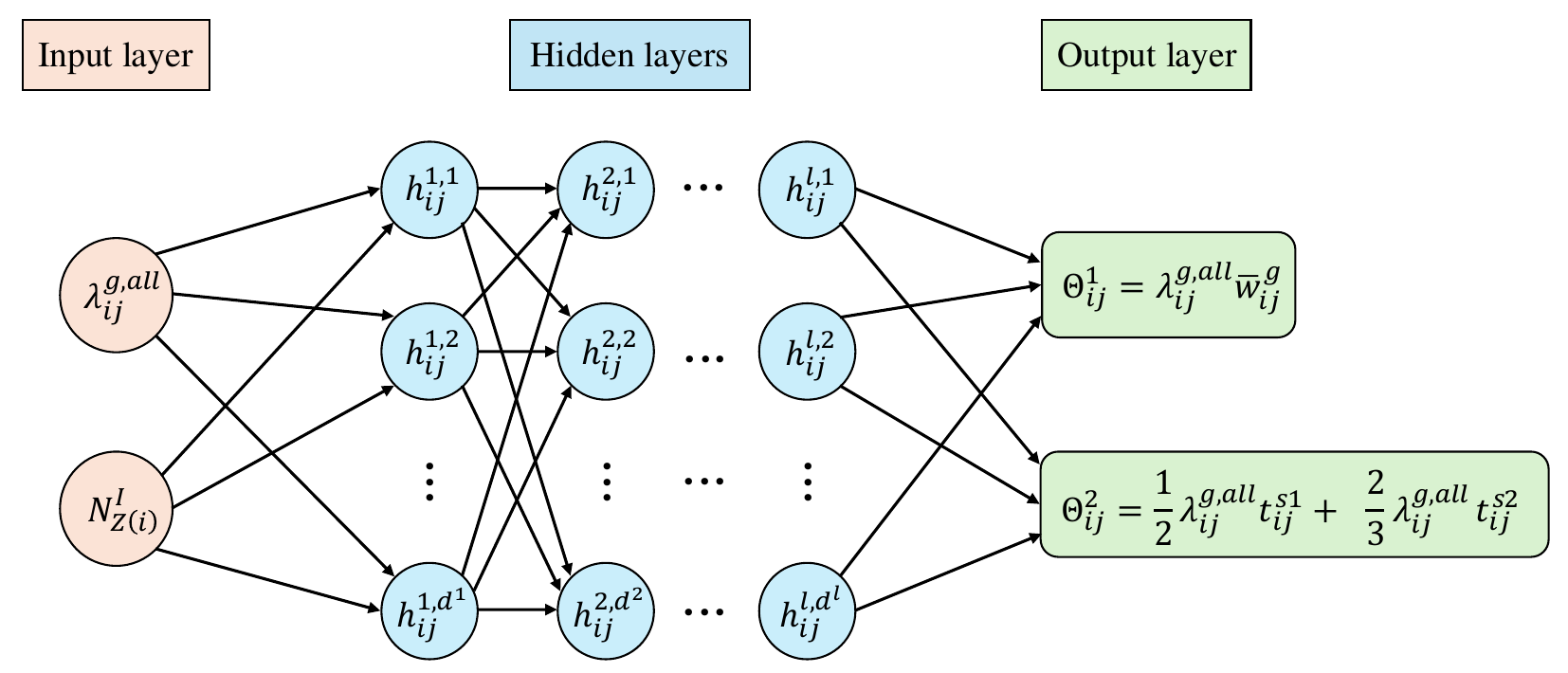}
     \caption{An example of a neural network with one hidden layer estimating the mapping in Proposition \ref{prop:MILP}.}
     \label{fig:NN}
 \end{figure}
Mathematically, it can be expressed as:
\begin{align}
    &\bm{h}_{ij}^{1} = ReLU(W_{ij}^{0}{[\lambda_{ij}^{g,all},N_{Z(i)}^I]}^T+\bm{b}_{ij}^{0}), \\
    &\bm{h}_{ij}^{m+1} = ReLU(W_{ij}^{m}\bm{h}_{ij}^{m}+\bm{b}_{ij}^{m}),\quad m = 1,\dots,l-1, \\
    &{[\bm{\Theta}_{ij}^1,\bm{\Theta}_{ij}^2]}^T = W_{ij}^{l}\bm{h}_{ij}^{l}+\bm{b}_{ij}^{l},
\end{align}
Note that the above relations are nonlinear due to the $ReLU$ function. In this case, the $v^{th}$ hidden neuron denoted by $h_{ij}^{m,v}$ is represented by
\begin{align}\label{v_hiddle_neuron_m}
    h_{ij}^{m,v} = \max\left\{\left(\sum_{u=1}^{d^{m-1}}w_{ij}^{m-1,uv}h_{ij}^{m-1,u}+b_{ij}^{m-1,v}\right),0\right\}.
\end{align}
To address the above issue, we introduce a continuous auxiliary variable $\check h_{ij}^{m,v}$ and a binary auxiliary variable $\kappa_{ij}^{m,v}$ to guarantee the non-negativeness of $h_{ij}^{m,v}$, then the $ReLU$ function (\ref{v_hiddle_neuron_m}) is equivalent to the following mixed-integer linear constraints \cite{patel2022neur2sp}:
\begin{align}
    &\sum_{u=1}^{d^{m-1}} w_{ij}^{m-1,uv} h_{ij}^{m-1,u} + b_{ij}^{m-1,v}  = h_{ij}^{m,v}-\check h_{ij}^{m,v}, \\
    &h_{ij}^{m,u} - M(1-\kappa_{ij}^{m,u}) \le 0, \\
    &\check h_{ij}^{m,u} - M \kappa_{ij}^{m,u} \le 0,\\
    &h_j^{m},\check h_j^{m} \ge 0,\\
    &\kappa_j^{m} \in\{0,1\},
\end{align}
where $M$ is a sufficiently large number. It indicates that when $\kappa_{ij}^{m,u} = 0$, we have $h_{ij}^{m,v}\ge 0$ and $\check h_{ij}^{m,v} = 0$; when $\kappa_{ij}^{m,u} = 1$, we have $h_{ij}^{m,v}= 0$ and $\check h_{ij}^{m,v} \ge 0$, which is equivalent to eh $ReLU$ constraints in (\ref{v_hiddle_neuron_m}).
Therefore, by linearizing all nonlinear terms in optimization problem (\ref{objective_function}) as per the aforementioned derivations and preserving the remaining mixed-integer linear constraints in (\ref{constraint}) as (\ref{constraint_remain_MILP}), we establish that the optimization problem (\ref{objective_MILP}) is equivalent to optimization (\ref{objective_function}) under neural-network estimation. This completes the proof. 
\end{proof}
    
The proposed method can notably enhance both the computational efficiency and solution quality by converting the intractable MINLP into MILP, which can be solved to global optimal solution in acceptable computation time by well-established algorithms such as branch-and-cut. The solution method is summarized in Algorithm \ref{algo:1}.
We emphasize that the effectiveness of the proposed methodological framework critically depends on the approximation provided by the neural network, and thus, the quality of the solutions obtained also crucially hinges on the accuracy of these neural network estimations. However, we argue that in the context of our problem, the accuracy of the neural network can be exceptionally high. This high level of accuracy is feasible because the neural network in question has only two inputs and two outputs, which allows it to be trained to very high precision using a relatively small dataset. Furthermore, a single neural network model serves distinct origin-destination pairs, differing only in the inputs and outputs, indicating that the network requires training only once. To support our argument, we have demonstrated the empirical estimation accuracy of this neural network in the case study using real data.
 \begin{algorithm}[t!] 
\caption{Solution algorithm to the optimization problem} \label{algo:1}
\hspace*{0.02in} {\bf Input:} 
model parameters \\
\hspace*{0.02in} {\bf Output:} 
optimal decisions of the platform $\bm{y}$,$\bm{z}$,$\bm{\hat \lambda}^a$,$\bm{\tilde\lambda}^a$,$\bm{q}^{I,r}$,$\bm{N}^I$, 
\begin{algorithmic}[1]
\STATE Data sampling: For each OD pair $(i,j)\in\mathcal{D}$, specify the range of $\lambda_{ij}^{g,all}$ and $N_{Z(i)}^I$ based on the model parameters, and then sample $N_s$ data points $(\lambda, n)$ uniformly from a rectangular area defined by specified ranges for $\lambda_{ij}^{g,all}$ and $N_{Z(i)}^I$. Denote the set of the sampled data by $\mathcal{S}_{ij}^{in}$, where $|\mathcal{S}_{ij}^{in}| = N_s$.
\STATE Label generation: For each OD pair $(i,j)$, generate dataset of input-output pairs $\mathcal{S}_{ij}$ for any input $(\lambda_{ij}^{g,all},N_{Z(i)}^I)$ in $\mathcal{S}_{in}$ and output $(\lambda_{ij}^{g,all}\bar w_{ij}^g,\frac{1}{2} \lambda_{ij}^{g,all}t_{ij}^{s1}+\frac{2}{3} \lambda_{ij}^{g,all}t_{ij}^{s_2})$ obtained by (\ref{eq:pickup_time})-(\ref{eq:t_s2}).
\STATE Train the neural network $\Phi_{ij}$ for each OD pair given the dataset $\mathcal{S}_{ij}$, obtaining the weights $\bm{W}_{ij}$ and bias $\bm{b}_{ij}$.
\STATE Solving the MILP (\ref{objective_MILP}) by branch-and-cut algorithm.
\RETURN $\bm{y}$,$\bm{z}$,$\bm{\hat \lambda}^a$,$\bm{\tilde\lambda}^a$,$\bm{q}^{I,r}$,$\bm{N}^I$.
\end{algorithmic}
\end{algorithm}

\section{Case Study}
To evaluate the performance of the proposed mathematical model and algorithm, we conduct a case study of the food-delivery market  on the real transportation network in Hong Kong. In particular, we consider a region with for centers: Po Lam, Tseung Kwan O, Hang Hau and Lohas Park, referred to as TKO area, with a transportation network of 66 nodes and 80 links, as shown in Figure \ref{fig:map}. 
\begin{figure}[htb!]
    \centering
    \includegraphics[width=\linewidth]{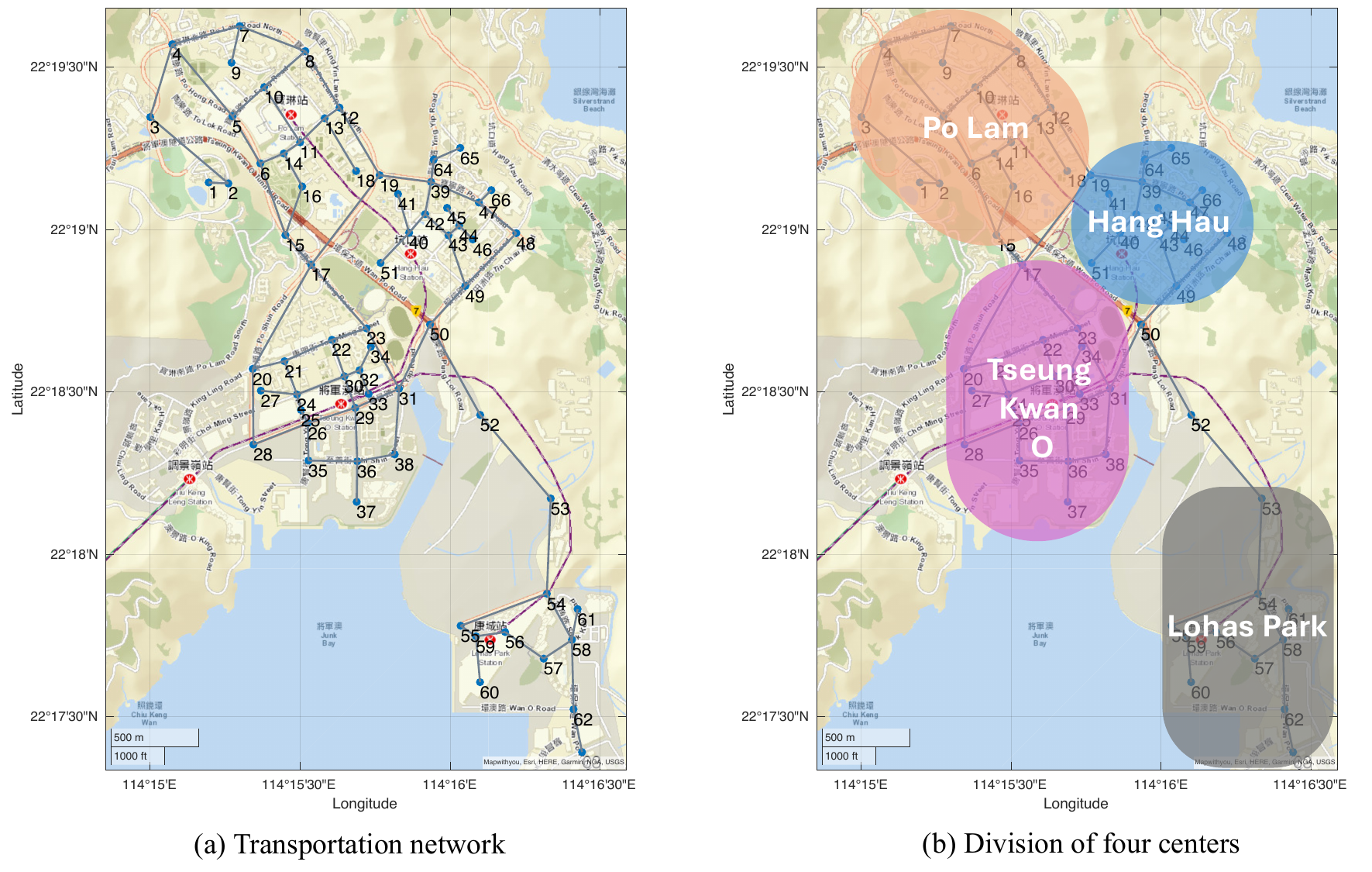}
    \caption{Transportation network of TKO area.}
    \label{fig:map}
\end{figure}
Based on the aforementioned network, we evaluate the proposed mathematical model and algorithm based on a real-world food-delivery dataset sourced from a proprietary consumer panel in Hong Kong, including 400,000 food-delivery transactions with GPS coordinates for pickup and drop-off points. Specifically focusing on orders within the TKO area, we map these locations onto the transportation network to derive the demand distribution for individual OD pairs across the network. Through an analysis of the food-delivery dataset and removal of the OD pair with relatively low demand, we identify 104 OD pairs representing orders dispatched from restaurant  (origin) to customer addresses (destination). The food-delivery demand ratio for each OD pair is then calculated based on the total order count in the dataset.
The arrival rate of food-delivery orders with origin and destination in the four areas (Po Lam, Tseung Kwan O, Hang Hau and Lohas Park) are illustrated in Figure \ref{fig:demand_distribution}.
The ground travel time for each link is calculated by dividing the real distance by a constant speed parameter $v_g$, while the air travel time is determined by dividing the straight-line distance between two nodes by the drone's  speed parameter $v_a$. In this study, we assume $v_g = 200$ m/s and $v_a = 800$ m/s. The average wage for a driver is fixed at $q=13$ \$/hr\footnote{Hereafter, we denote \$ as US\$.}, which aligns with the average salary of $100$ HK\$/hr for food-delivery couriers in Hong Kong\footnote{The average salary is based on estimates provided by \url{https://www.glassdoor.com.hk/Salaries/food-delivery-salary-SRCH_KO0,13.htm}.}. According to research by \cite{doole2020estimation}, the average total daily costs (comprising both investment and operational expenses) in a realistic scenario are approximately 198,016 Euros for 18,458 drones.
Therefore, we set the average investment and operational costs for a drone to $C^a = 1.67$ \$/hr. 
Given the transportation network of TKO area, we choose four potential locations for launchpad, i.e. $\mathcal{L} = \{11,29,42,56\}$\footnote{In the considered area of limited size, it does not make sense to build too many launchpads for drones (similar to how a city normally does not need too many airports). However, it is important to emphasize that the proposed algorithm can easily scale to accommodate a larger number of potential launchpads. }. These locations are strategically situated at the hearts of Po Lam, Tseung Kwan O, Hang Hau, and Lohas Park. Additionally, we choose 9 potential locations for kiosks,  $\mathcal{K} = \{9,10,13,25,28,36,40,42,56\}$, ensuring comprehensive coverage across the area with high concentration of delivery destinations.
\begin{figure}[htb!]
    \centering
    \includegraphics[width=0.7\linewidth]{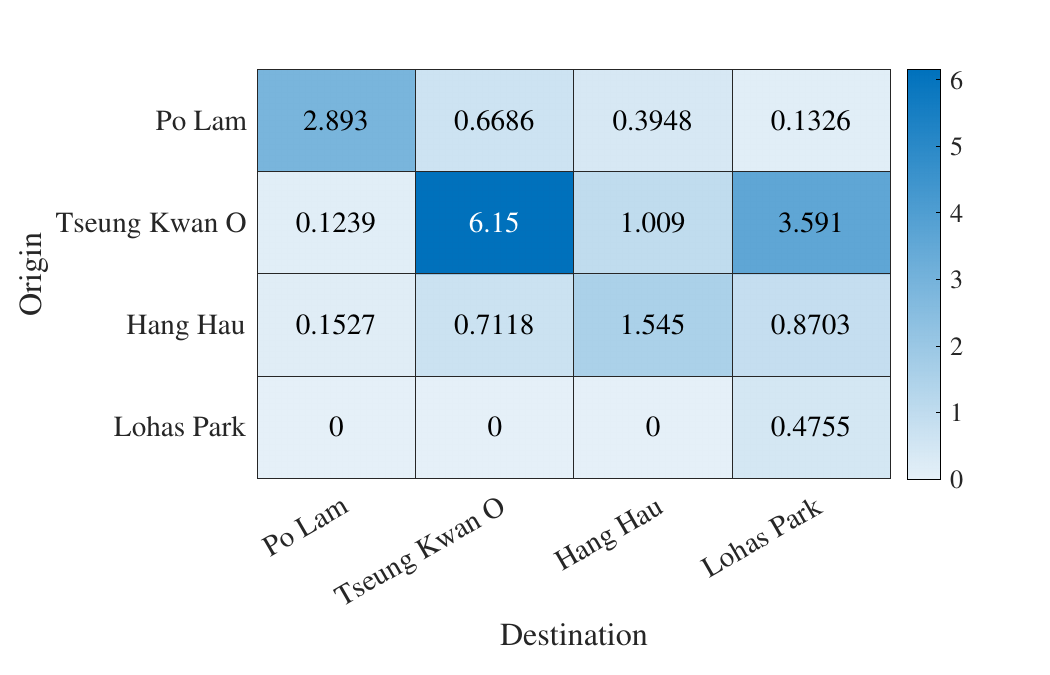}
    \caption{The arrival rate (/min) of food-delivery orders with origin and destination in the four centers.}
    \label{fig:demand_distribution}
\end{figure}
To assess the effectiveness of the proposed learning-assisted solution algorithm, we detail the training procedure and show the neural network's performance in Section \ref{sec:casestudy_NN}.
In order to analyze the expansion of air delivery services over reduced infrastructure expenses and increasing demand, as well as to explore the influence of drone delivery services on the food-delivery market, we examine the infrastructure deployment and market outcomes across varying construction costs and levels of food-delivery demand. The findings of these investigations are presented in Sections \ref{sec:casestudy_impact_cost} and \ref{sec:casestudy_impact_demand}, respectively.

\subsection{Experimental Setting and Computational Performance}\label{sec:casestudy_NN}

Considering the transportation network and the potential locations of launchpads and kiosks, the number of overall ground flow OD pairs (including orders delivered by ground couriers from a restaurant to a customer location, from a restaurant to a launchpad, and from a kiosk to a customer location) is $|\mathcal{D}|=277$. However, the input-output relations of all the OD paris are identical, although the input values and output values differ. In this case, we can employ a single neural network to approximate a 2-dimensional input $\mathbf{x}$ to a 2-dimensional output task $\mathbf{\Theta}$ for all OD pairs $(i,j)$ in the set $\mathcal{D}$, as depicted in Figure \ref{fig:NN}. For each OD pair, we generate 45,000 data points through grid sampling and standardization, which involves dividing a rectangular area defined by the input range into multiple grids, extracting data points $\mathbf{x}_{ij}$ at the grid vertices, and pre-processing them to standardized values with a mean of 0 and a standard deviation of 1. 
The neural network comprises one hidden layer with 12 neurons.
The proposed model is implemented using PyTorch on a workstation equipped with a GeForce RTX 4090 GPU. The learning rates, $\eta$, is set at 0.001, and the training process spans 300 epochs. The loss function we used is MSE defined by:
\begin{equation} 
{L= \frac{1}{2} \|\mathbf{\Theta}_k - \mathbf{\hat{\Theta}}_k\|_2},
\label{Equ: Loss}
\end{equation}
where $\mathbf{\Theta}_k$ and $\mathbf{\hat{\Theta}}_k$ refer to actual and predicted values. The dataset is split in a 9:1 ratio, where 90\% of the data is used for training and the remaining 10\% for testing. 
In the training process, performance analysis is performed using four key metrics: mean absolute error (MAE), root mean squared error (RMSE), MAPE, and R-squared score (R2). The above metrics can be specified as follows:
\begin{equation}
{\mathrm{MAE} = \frac{1}{K} \sum_{k=1}^{K}| \mathbf {\Theta}_{k} -  \mathbf {\hat{\Theta}}_{k}|}
\end{equation}
\begin{equation}
{\mathrm{RMSE} = \sqrt{\frac{1}{K} \sum_{k=1}^{K}\|\mathbf{\Theta}_k - \mathbf{\hat{\Theta}}_k\|_2}}
\label{RMSE}
\end{equation}
\begin{equation}
{\mathrm{R2} = 1 - \frac{\sum_{k=1}^{K}\|\mathbf{\Theta}_k - \mathbf{\hat{\Theta}}_k\|_2}{\sum_{k=1}^{K}\|\mathbf{\Theta}_k - \mathbf{\bar{\Theta}}\|_2}}
\end{equation}
where
$K$ is the number of data samples;
$\mathbf{\bar{\Theta}}$ is the average of output labels.
The mean, minimum, and maximum values of the metrics for all 277 cases across the entire dataset are consolidated in Table \ref{tab:train_result}. The results demonstrate that the neural network's estimation on the testing dataset is highly accurate, exhibiting minor errors for both outputs, with $R^2$ value greater than 0.999.

\begin{table}[H]
\centering
\begin{tabular}{ccccc} 
\hline\hline                                                           Label                & Statistic & MAE       & RMSE       & $R^2$                \\ 
\hline
\multirow{3}{*}{1st} & mean      & 0.007041~ & 0.009827~  & 0.999987\\
                      & min       & 0.000357~ & 0.000493~  & 0.999873~  \\
                      & max       & 0.041356~ & 0.056011~  & 0.999999~  \\
                      \hline
 \multirow{3}{*}{2nd} & mean      & 0.004688~ & 0.006705~  & 0.999972~   \\
                      & min       & 0.000058~ & 0.000105~  & 0.999791~\\
                      & max       & 0.026009~ & 0.036524~  & 0.999998~   \\
\hline\hline
\label{Tab1}
\end{tabular}
\caption{Statistical results of the training process across all OD pairs.}\label{tab:train_result}
\end{table}

To evaluate the accuracy of the neural network estimation in the optimization problem (\ref{objective_MILP}), we choose $C_l = 55$ \$/hr, $C_k = 12$ \$/r and $\max_{i,j}\lambda_{ij} = 0.3$ /min, solve MILP (\ref{objective_MILP}), and compare the predicted values of  both outputs from neural networks and actual values of  both outputs from original nonlinear constraints for all 277 cases under the optimal solution to MILP (\ref{objective_MILP}). The MILP is solved by Gurobi in MATLAB on a Dell Desktop machine with 6-core 3.10 GHz i5–10 500 CPU.  The comparison results are depicted in Figure \ref{fig:predict}, illustrating that the predicted outputs closely align with the true values under the optimal solution.

\begin{figure}[H]
     \centering
     \begin{subfigure}[The actual value and predicted value of the first output for all OD pairs.]{
         \centering
         \includegraphics[width=0.45\textwidth]{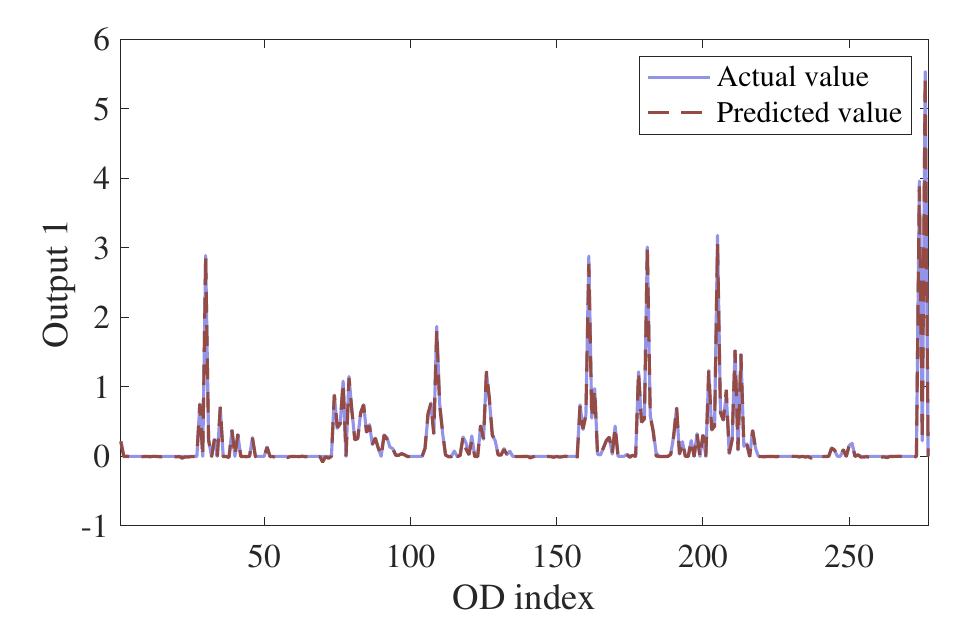}
         \label{fig:output_1_predict}}
     \end{subfigure}
     \begin{subfigure}[The actual value and predicted value of the second output for all OD pairs.]{
         \centering
         \includegraphics[width=0.45\textwidth]{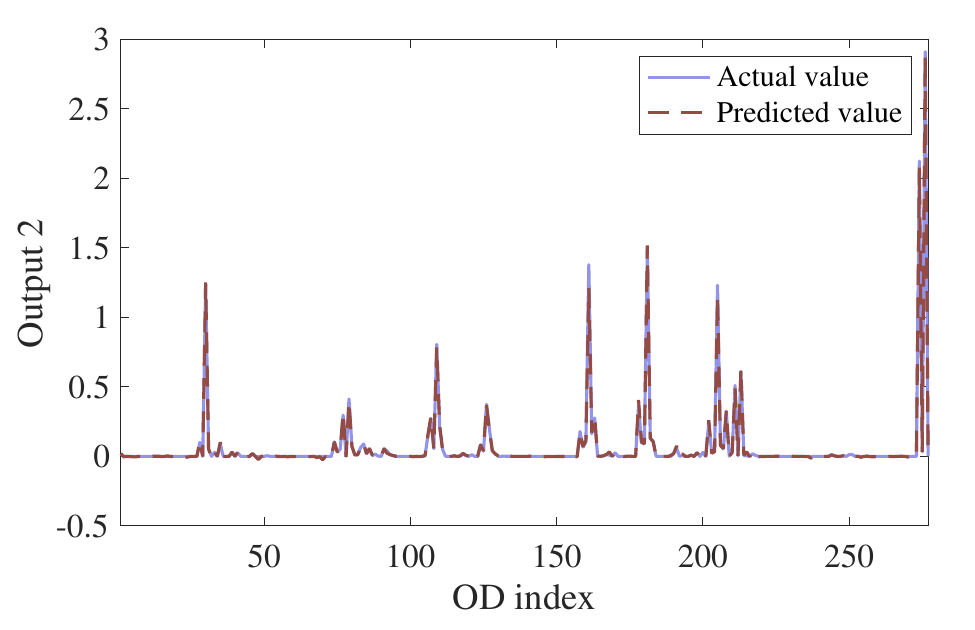}
         \label{fig:output_2_predict}}
     \end{subfigure}
       \caption{The comparison of actual value and predicted value of two outputs for all OD pairs.}
        \label{fig:predict}
\end{figure}

\subsection{Impacts of Construction Costs}\label{sec:casestudy_impact_cost}
To evaluate how the platform expands the air delivery services as the  construction cost of launchpad/kiosk reduces, and how this will influence the overall food-delivery market, we fixed the food-delivery demand $\lambda_{ij}$ by normalizing the demand to $\max_{i,j} \lambda_{ij} = 0.3$ while keeping the ratio of demand for other OD pair to $\max_{i,j} \lambda_{ij}$ unchanged, gradually reduce the construction costs of launchpads and kiosks, and investigate the impacts of this change on platform, food-delivery customers, and couriers. We consider 16 cases of launchpad/kiosk construction costs, as summarized  in Table \ref{tab:cost}. Under each case, the optimal decisions on the locations of launchpads and kiosks are also presented in Table \ref{tab:cost}.
The platform costs (including the construction costs for launchpads and kiosks, and the operational costs for the mixed fleet of couriers and drones), average food-delivery order delivery time, the average pooling probability over all the food-delivery orders, and the courier fleet size under different cases in Table \ref{tab:cost} are summarized in Figure \ref{fig:bar}.

\begin{table}[H]
    \centering
    \begin{tabular}{c|cccc}
    \hline\hline
       Index  & Launchpad costs & Kiosk costs & Built launchpad locations & Built kiosk locations \\
       \hline
       1 & 90& 19.6 & - & -\\
       2 & 85& 18.5 & 29 & 56\\
       3 & 80 & 17.4 & 29 & 56 \\
       4 & 75 &  16.4 & 29 & 40, 56\\
       5 & 70 & 15.3 & 29 & 40, 56\\
       6 & 65 & 14.2 & 29 & 40, 56\\
       7 & 60 & 13.1 & 29 & 40, 56\\
       8 & 55 & 12 & 29 & 40, 56\\
       9 & 50 & 10.9 & 29 & 40, 56\\
       10 & 45 & 9.8 & 29 & 40, 56\\
       11 & 40 & 8.7 & 29 & 40, 56 \\
       12 & 35 & 7.6 & 29 & 40, 56\\
       13 & 30 & 6.5 & 29 & 40, 56\\
       14 & 25 & 5.4 & 29, 42 & 25, 40, 56\\
       15 & 20 & 4.4 & 29, 42 & 25, 40, 56\\
       16 & 15 & 3.3 & 29, 42 & 10, 25, 40, 56 \\
        \hline\hline
    \end{tabular}
    \caption{Construction costs (\$/hr) for a launchpad and a kiosk and the locations under different cases.}
    \label{tab:cost}
\end{table}

\begin{figure}[H]
     \centering
     \begin{subfigure}[The platform costs (\$/min) under different construction costs.]{
         \centering
         \includegraphics[width=0.45\textwidth]{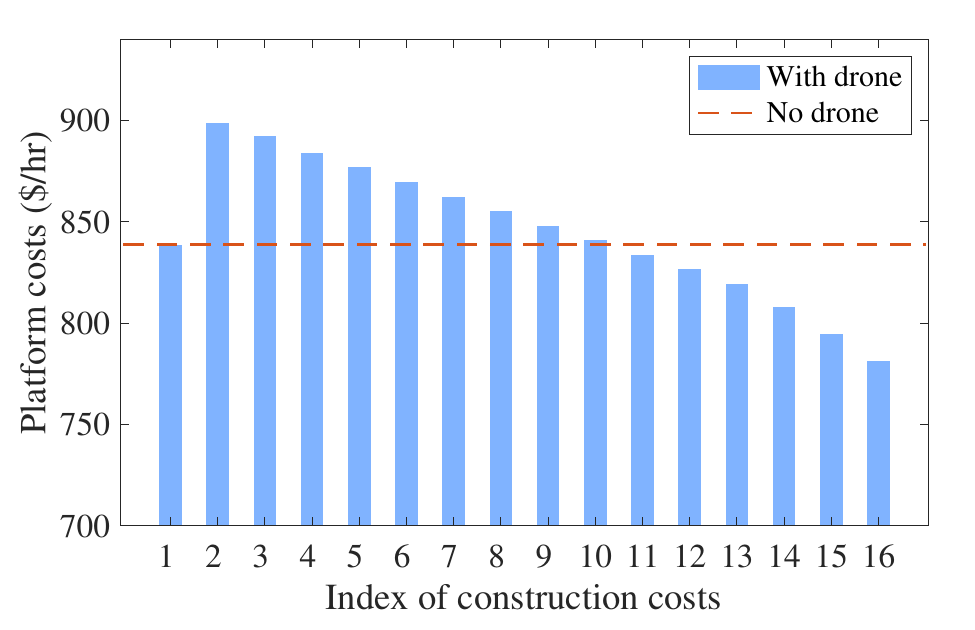}
         \label{fig:bar_cost}}
     \end{subfigure}
     \begin{subfigure}[The average delivery time (min) under different construction costs.]{
         \centering
         \includegraphics[width=0.45\textwidth]{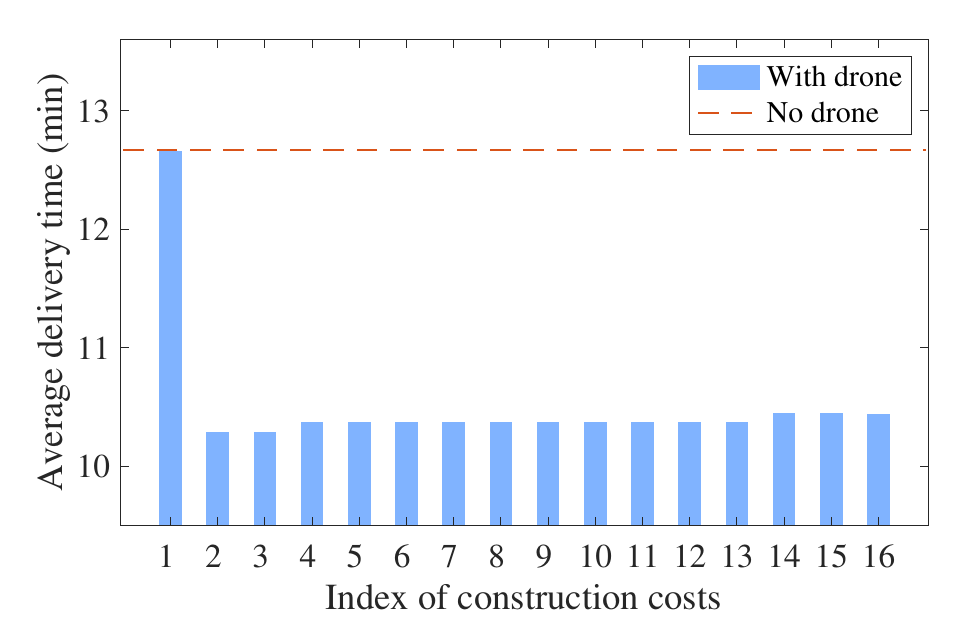}
         \label{fig:bar_time}}
     \end{subfigure}
     \begin{subfigure}[The average probability of pooling for orders under different construction costs.]{
         \centering
         \includegraphics[width=0.45\textwidth]{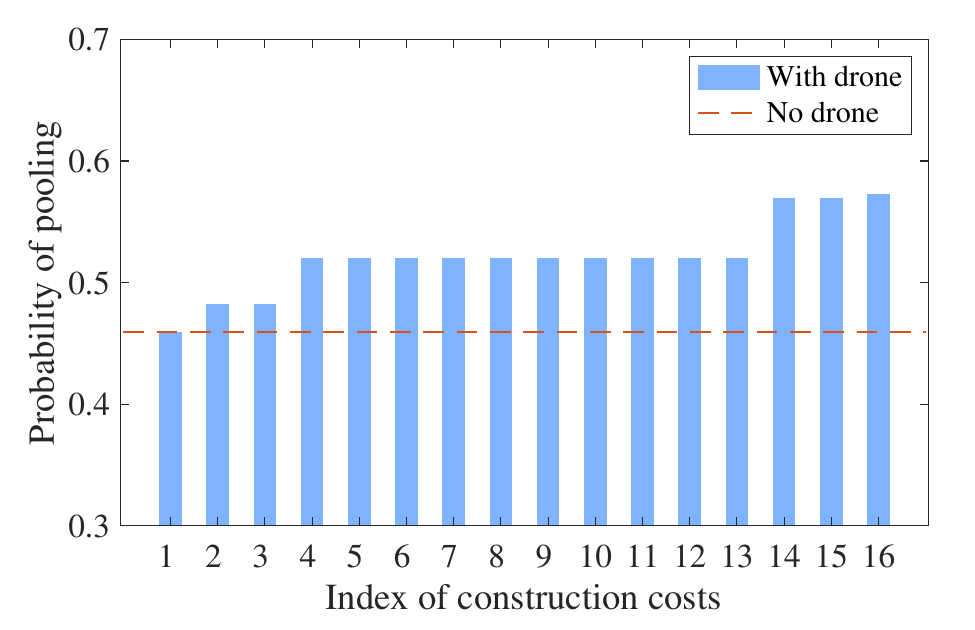}
         \label{fig:bar_pool_prob}}
     \end{subfigure}
     \begin{subfigure}[The number of couriers under different construction costs.]{
         \centering
         \includegraphics[width=0.45\textwidth]{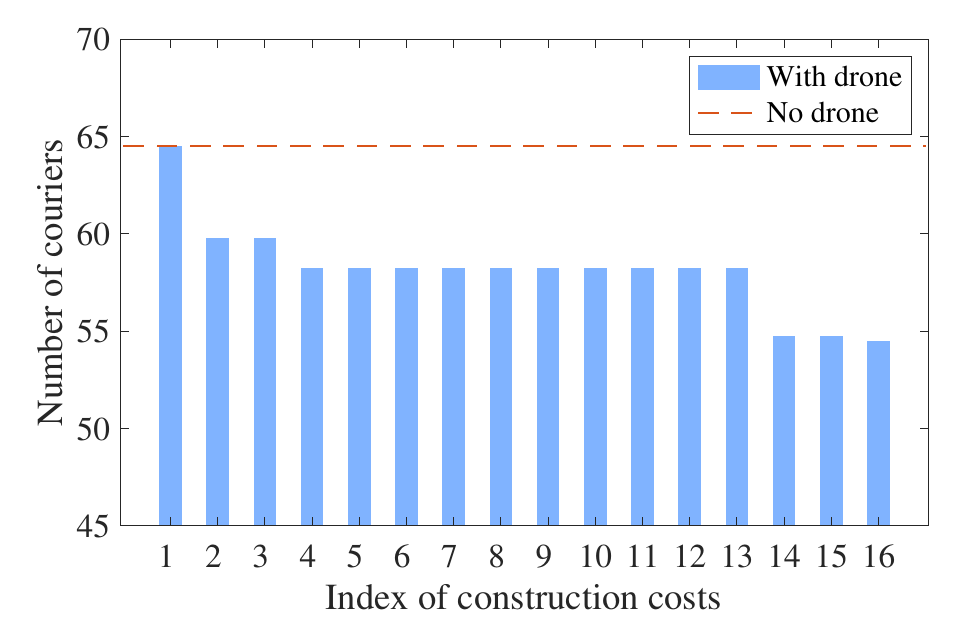}
         \label{fig:bar_courier}}
     \end{subfigure}
        \caption{The platform costs, average delivery time, pooling probability and number of couriers under different construction costs.}
        \label{fig:bar}
\end{figure}

It is evident that as the construction costs for launchpads and kiosks decrease, the food-delivery platform will expand its air delivery network by constructing more launchpads and kiosks, as indicated in Table~\ref{tab:cost}. This expansion leads to a reduced courier fleet size, as shown in Figure \ref{fig:bar_courier}, and an increased probability for order bundling, as depicted in Figure \ref{fig:bar_pool_prob}, compared to scenarios without air delivery services. These results are straightforward: the augmented deployment of launchpads and kiosks facilitates the opening of additional air routes, enabling drones to fulfill more orders and thereby replace human couriers, subsequently reducing the courier fleet size. Furthermore, the introduction of air delivery routes significantly enhances flexibility in route planning. This allows the platform to deliver not only directly to the destinations but also to strategically drop orders at launchpads. Consequently, both the origin and destination of the trips become flexible variables, rather than exogenously fixed parameters, enhancing the potential for order pooling. For example, the addition of more air delivery routes allows for the bundling of orders with the same origin but different destinations if they utilize air delivery and are assigned to the same launchpad. Similarly, food-delivery orders originating from different restaurant locations but destined for the same endpoint can also be bundled if they are serviced by air delivery utilizing the same kiosks. Thus, the expansion of air delivery routes increases the average probability of order bundling, leading to better utilization of couriers and reducing the miles they travel.

One interesting observation from our analysis is that the total cost to the platform, which includes the infrastructure costs of launchpads and kiosks, as well as the operational costs for the mixed fleet of couriers and drones (excluding the last term in equation (\ref{objective_MILP})), does not exhibit a monotonic relationship with the construction costs, as depicted in Figure \ref{fig:bar_cost}. Specifically, the costs in case 2 are notably higher than those in case 1. However, as the construction costs decrease further, the overall cost tends to decrease. To understand this result, we note that when construction costs are prohibitively high, the platform chooses not to construct any launchpads or kiosks, recognizing that the potential reductions in customer waiting times and operational costs do not justify the substantial upfront investment in infrastructure. However, when construction costs are reduced but still relatively high (as in cases 2, detailed in Table \ref{tab:cost}), the platform may decide to construct infrastructure selectively. For instance, it might activate air routes for high-demand, long-distance OD pairs, such as from Tseung Kwan O to Lohas Park, using launchpad 29 and kiosk 56. In these scenarios, despite the high construction costs, the significant reduction in delivery times through faster air transport (especially for long-distance ODs) justifies the upfront investment in infrastructure. As construction costs continue to decrease, the platform expands its infrastructure to include air routes for shorter-distance OD pairs, such as from Tseung Kwan O to Hang Hau using launchpad 29 and kiosk 40. In these cases, the cost per unit of launchpad and kiosk construction (already reduced) is outweighted by the savings in operational cost it induced from reduced courier usage, result in a net reduction in total costs, even as new facilities are constructed.

\begin{table}[t]
    \centering
    \begin{tabular}{c|cc}
    \hline\hline
      & Tseung Kwan O to Lohas Park & Tseung Kwan O to Hang Hau\\
    \hline 
    Avg air delivery time (min) & 20.5 & 12.9 \\
    Avg ground delivery time (min) & 15.3 & 14.2 \\
    Air delivery demand (/min) & 1.07 & 0.30 \\
    Ground delivery demand (/min) & 0.12 & 0.03\\
         \hline\hline
    \end{tabular}
    \caption{Average delivery time and demand for air delivery services and ground delivery services with origin in TKO and destination in Lohas Park/Hang Hau under the construction cost case 5 in Table \ref{tab:cost}.}
    \label{tab:time}
\end{table}

Another interesting, albeit counterintuitive, observation is that the expansion of air delivery services does not necessarily entail reduced delivery times for customers, despite the significantly higher speeds of air compared to ground delivery. As depicted in Figure \ref{fig:bar_time}, the introduction of the first launchpad-kiosk pair between Tseung Kwan O and Lohas Park significantly reduces the average delivery time from case 1 to case 2 for long-distance orders. However, as construction costs decrease further and the platform activates additional air routes, the delivery times actually increase, despite the improved availability of air services. This  highlights a crucial trade-off between the travel time savings induced by the faster air delivery and the associated detours and extra waiting times at launchpads and kiosks. This trade-off appears to be heavily dependent on the distance of the order. For long-distance trips, the detours are minor compared to the substantial time savings from faster drone delivery. In contrast, for short-distance trips, the detours and additional waiting times can significantly extend the delivery time compared to ground delivery (which is already short due to the short distance), outweighing the time saved by the drones. Consequently, as the platform initially activates air delivery routes for long distances and subsequently expands them to shorter routes, the average delivery time first decreases (due to the opening of long-distance routes) and then increases (due to the opening of short-distance routes). To demonstrate the contrasting effects of air delivery services on the delivery times of long-distance OD and short-distance OD pairs, we analyze the average delivery times and demand for orders originating in Tseung Kwan O and destined for Lohas Park/Hang Hau, allocated to either air delivery services or ground delivery services under construction cost case 5, as detailed in Table \ref{tab:time}. It is illustrated that with a launchpad in Tseung Kwan O and two kiosks in Lohas Park and Hang Hau respectively, the platform allocates the majority of food-delivery orders from Tseung Kwan O to Lohas Park/Hang Hau to air delivery services.  The average delivery time for air delivery services from Tseung Kwan O to Lohas Park (representing a long distance) is shorter compared to ground delivery services, while the average delivery time for air delivery services from Tseung Kwan O to Hang Hau (representing a short distance) is longer than that of ground delivery services. This echoes our previous analysis: while air delivery can significantly speed up transport by replacing slower ground transportation, the additional waiting times and detours, especially for short-distance deliveries, at launchpads and kiosks for courier pickups can be substantial. In long-distance deliveries, the time saved by the drones' higher speed outweighs the extra waiting time incurred during the first and last legs of the delivery. Conversely, in short-distance deliveries, the additional delays dominate, resulting in longer overall delivery times.


\subsection{Impacts of Food-Delivery Demand}\label{sec:casestudy_impact_demand}
To explore how the rising demand for food delivery influences the expansion of air delivery services, we maintained constant construction costs for launchpads and kiosks while incrementally increasing the maximum demand across all origin-destination pairs from $0.1$/min to $1.1$/min. The decisions made by the platform regarding the placement of launchpads and kiosks at different demand levels are summarized in Table \ref{tab:lambda}. A comparative analysis of platform costs, average delivery time, average pooling probability for an order, and courier fleet size under various demand scenarios is presented in Figure \ref{fig:curve}. This analysis contrasts conventional delivery methods without drones against scenarios that incorporate drone delivery services, highlighting how air delivery impacts the platform, the food-delivery customers, and the couriers in the food-delivery market as demand fluctuates.

\begin{table}[t]
    \centering
    \begin{tabular}{c|ccc}
    \hline\hline
       Index  & Maximum arrival rate & Built launchpad locations & Built kiosk locations \\
       \hline
       1 & 0.1&  - & -\\
       2 & 0.3&  29 & 40, 56\\
       3 & 0.5 &  29, 42 & 25, 40, 56 \\
       4 & 0.7 &   29, 42 & 25, 40, 56 \\
       5 & 0.9 &  29, 42 & 10, 25, 40, 56\\
       6 & 1.1 & 29, 42 & 10, 25, 40, 56\\
        \hline\hline
    \end{tabular}
    \caption{The locations of launchpads and kiosks under different demand level.}
    \label{tab:lambda}
\end{table}

\begin{figure}[t]
     \centering
     \begin{subfigure}[The platform costs (\$/min) under different demand level.]{
         \centering
         \includegraphics[width=0.45\textwidth]{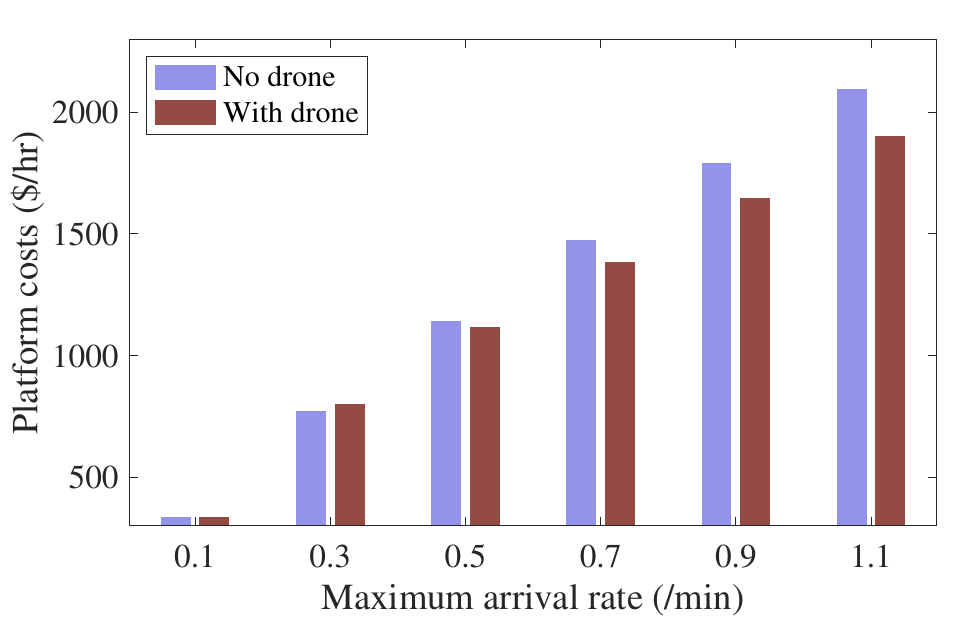}
         \label{fig:curve_cost_lambda}}
     \end{subfigure}
     \begin{subfigure}[The average delivery time (min) under different demand level.]{
         \centering
         \includegraphics[width=0.45\textwidth]{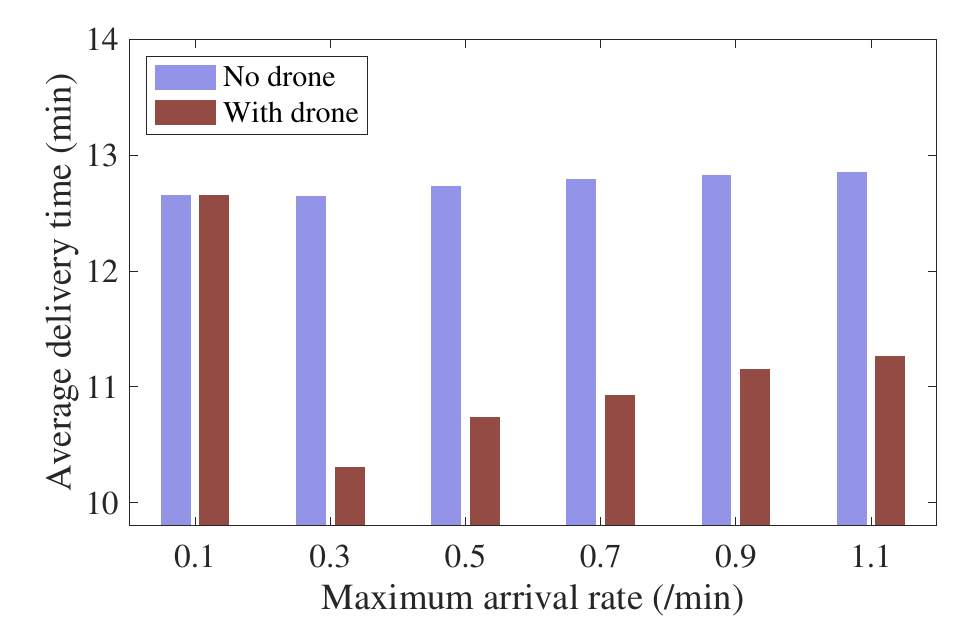}
         \label{fig:curve_time_lambda}}
     \end{subfigure}
     \begin{subfigure}[The average probability of pooling for orders under different demand level.]{
         \centering
         \includegraphics[width=0.45\textwidth]{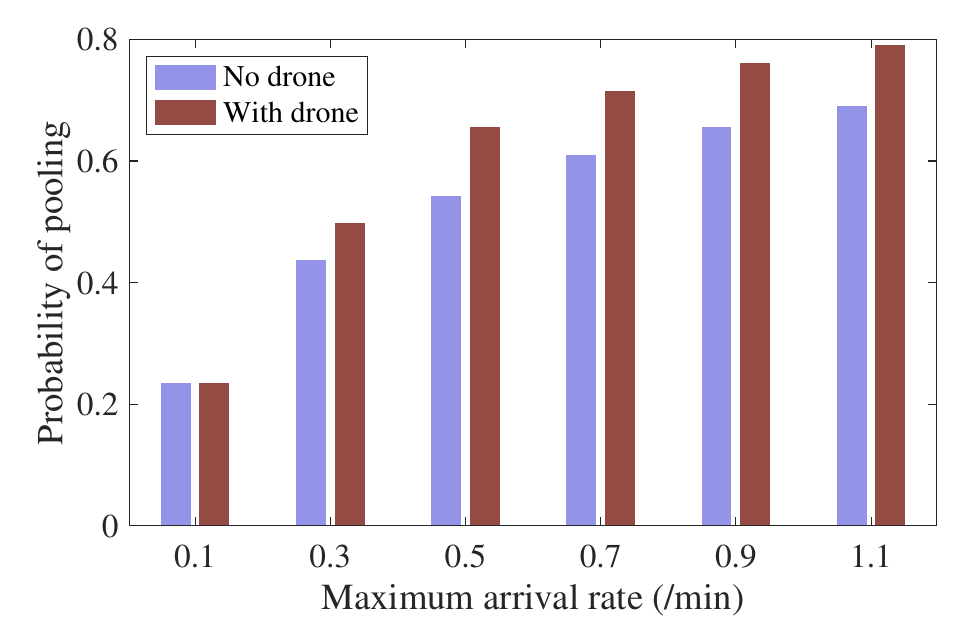}
         \label{fig:curve_pool_prob_lambda}}
     \end{subfigure}
     \begin{subfigure}[The number of couriers under different demand level.]{
         \centering
         \includegraphics[width=0.45\textwidth]{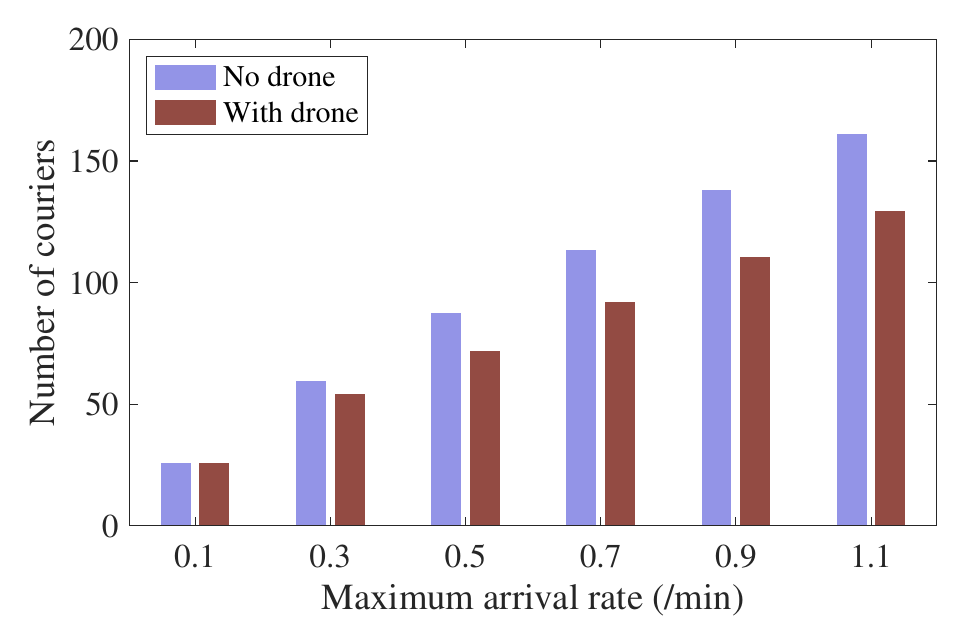}
         \label{fig:curve_courier_lambda}}
     \end{subfigure}
        \caption{The platform costs, average delivery time, pooling probability and number of couriers under demand level.}
        \label{fig:curve}
\end{figure}

As food-delivery demand escalates, it becomes evident that the platform will augment its air delivery network by constructing additional launchpads and kiosks, as indicated in Table \ref{tab:lambda}. This network expansion results in increased platform costs, encompassing both infrastructure and operational expenses, as depicted in Figure \ref{fig:curve_cost_lambda}. Correspondingly, there is a rise in the courier fleet size, shown in Figure \ref{fig:curve_pool_prob_lambda}, and an enhanced probability for order bundling, as illustrated in Figure \ref{fig:curve_courier_lambda}. These outcomes are intuitive: the surge in demand compels the platform to boost its supply capabilities, including both drones and couriers. Consequently, the number of launchpads and kiosks expands, as does the courier fleet. Moreover, the introduction of air delivery routes significantly improves flexibility in route planning, which elevates the average probability of order bundling. This leads to more efficient utilization of couriers and a reduction in the miles they travel, mirroring the findings discussed in the case studies in Section \ref{sec:casestudy_impact_cost}.

Interestingly, a comparison of platform costs, which encompass both infrastructure and operational expenses but exclude the last term in equation (\ref{objective_MILP}), reveals a nuanced impact of air delivery. When the maximum arrival rate is 0.3, the platform costs are higher in the air delivery scenario than in the benchmark case without drone delivery. However, when the maximum arrival rate is 0.5 or higher, the costs are lower than the benchmark, as illustrated in Figure \ref{fig:curve_cost_lambda}. This observation echoes the findings depicted in Figure \ref{fig:bar_cost}. Specifically, when the maximum demand increases from 0.1 to 0.3, the platform may opt for selective infrastructure construction. For example, it might activate air routes for high-demand, long-distance origin-destination pairs. In these scenarios, despite the high initial costs, the substantial reduction in delivery times through faster air transport—particularly for long-distance routes—justifies the upfront investment. As demand continues to increase beyond 0.5, the platform broadens its infrastructure to accommodate air routes for shorter-distance pairs. In these instances, the reduced unit costs of launchpad and kiosk construction are outweighed by the operational savings from decreased courier usage, resulting in a smaller total costs compared to the benchmark case, even as new facilities are constructed. Furthermore, additional insights are revealed in Figure \ref{fig:curve_time_lambda}. In the case with drones, delivery time is not a monotonic function of increasing demand. Specifically, as demand increases, delivery time first decreases significantly, and then gradually begins to increase. This pattern is also echoed in Figure \ref{fig:bar_time}, both of which highlight a crucial trade-off: the travel time savings induced by faster air delivery are sometimes counterbalanced by the associated detours and extra waiting times at launchpads and kiosks, especially for shorter-distanced trips.

\section{Conclusions}

This paper investigates the infrastructure location planning and order assignment problem for an on-demand food-delivery platform operating over a multimodal delivery network. Specifically, we consider a platform that offers both ground delivery services and drone-assisted air delivery services. The platform must determine the optimal locations for launchpads and kiosks within a transportation network under a constrained budget, and develop an order assignment strategy that effectively allocates food-delivery orders between ground and air delivery. These decisions are influenced by the bundling probabilities of ground deliveries and the waiting times at launchpads and kiosks for air deliveries. To address the former, we developed a steady-state equilibrium model that prescribes the matching process between couriers and orders. To characterize the latter, we formulated a double-ended queue model for the interactions at launchpads between orders and drones. The optimal decisions for the platform are framed as a MINLP problem aiming to minimize a trade-off between platform costs and total delivery time. To tackle the nonlinearity inherent in the MINLP formulation, we propose a learning-assisted solution algorithm, which transforms the MINLP into a MINP and enables efficient computation of the global optimum solution.

We validate the proposed mathematical model and algorithms via a case study in Hong Kong. Initially, we explore the impacts of the launchpad/kiosk construction costs on the multimodal food-delivery market. Our findings indicate that as the infrastructure costs for launchpads and kiosks decrease, the food-delivery platform opts to expand air delivery services across the transportation network, starting with long-distance routes and gradually extending to shorter routes. This expansion leads to reduced operational costs for the platform, a smaller courier fleet size, and increased opportunities for order bundling. However, expanding air delivery services does not necessarily result in reduced delivery times for customers, despite the significantly higher speeds of air compared to ground delivery. This highlights a key trade-off between the travel time savings induced by faster air delivery and the associated detours and additional waiting times at launchpads and kiosks. We discover that this trade-off is heavily dependent on the distance of the order. For long-distance trips, the detours are minor compared to the substantial time savings from faster drone delivery. In contrast, for short-distance trips, the detours and additional waiting times can significantly extend the delivery time compared to ground delivery, outweighing the time saved by the drones. Consequently, as the platform initially activates air delivery routes for long distances and later expands them to shorter routes, the average delivery time first decreases and then increases. Next, we investigate the effects of growing demand on the service expansion. Similar results were observed, and we have carefully analyzed them.

This study presents several opportunities for expansion. Firstly, while this work considers exogenous demand and a sufficiently large pool of courier supply, exploring optimal pricing strategies and payment strategies under elastic supply and demand could be an interesting direction. Secondly, enhancing the mathematical model to include advanced matching mechanisms and flexible bundling sizes could unlock additional benefits for the drone-assisted delivery services.

\begin{appendices}
\newpage

\section{Notations}\label{append:notation}
\begin{center}
	\begin{longtable}{p{2.4cm}p{14cm}}
		\hline\noalign{\smallskip}
        \hline
		Notation & Definition \\
    \hline
    \multirow{2}*{\bf{Set}} & \\
    ~ &  \\
    $\mathcal{N}$ & The set of nodes\\
    $\mathcal{R}$ & The set of restaurant locations \\
    $\mathcal{L}$ & The set of potential launchpad locations \\
    $\mathcal{K}$ & The set of potential kiosk locations \\
    $\mathcal{C}$ & The set of customer locations\\
    $\mathcal{Z}$ & The set of zones \\
    $\mathcal{D}$ & The set of OD pairs for all ground flow \\
    \multirow{2}*{\bf{Parameters}} & \\
    ~ &  \\
    $\lambda_{ij}$ & Customer demand from restaurant $i$ to customer location $j$ \\
    $\gamma_r$& Reliability level at launchpads\\
    $t_i^c$ & Average time for a courier to pick up the second/third order \\
    $t_{ij}$ & Average travel time from node $i$ to node $j$ \\
    $t_{Z_1Z_2}^r$ & Average repositioning time for idle couriers from zone $Z_1$ to zone $Z_2$ \\
    $t_{lk}^a$ & Average time for a drone to fly from location $l$ to location $k$ \\
    $\bar t_d$/$\bar t_r$ & Maximum allowable flight time for drone delivery and relocation \\
    $C_L^l$/$C_K^k$ & Average construction costs for a launchpad/kiosk per unit of time\\
    $C^a$ & Average investment and operational costs of a drone per unit of time\\
    $q$ & Average wage paid to a courier per unit of time \\
    $\alpha_w$ & Trade-off parameter between platform costs and customer delivery time \\
    $H$ & Scaling parameter on pickup time \\
    $M$ & Capacity of waiting orders at each launchpad \\
    $M_L^l$, $M_K^k$, $M_{\gamma}$, $M_{nn}$ & Sufficiently large numbers\\
    $\bm{W}_{ij}^m$ & The weights of the neural network\\
    $\bm{b}_{ij}^m$ & The bias of the neural network \\
    \multirow{2}*{\bf{Variables}} & \\
    ~ &  \\
    $\lambda_{ilkj}^a$ & Realized arrival rate of air delivery orders with origin $i$, destination $j$, launchpad $l$ and kiosk $k$ \\
    $\hat \lambda_{ilkj}^a$ & Designed arrival rate of air delivery orders with origin $i$, destination $j$, launchpad $l$ and kiosk $k$ \\
    $\lambda_{ij}^g$ & Arrival rate of orders with origin $i$ and destination $j$ assigned to ground delivery services \\
    $y_l$/$z_k$ & Binary variable indicating whether a launchpad/kiosk is built at location $l$/$k$ \\
    $\lambda_{ij}^{g,all}$ & Overall ground flow for both ground delivery services and air delivery services \\
    $\tilde \lambda_{kl}^a$ & Repositioning flow of drones from kiosk $k$ to launchpad $l$ \\
    $\nu_l^o$/$\nu_l^d$ & Arrival rate of orders/drones at launchpad $l$\\
    $\bar N_l^o$/$\bar N_l^d$ & Average number of drones/orders waiting at launchpad $l$ \\
    $\gamma_l$ & Reliability level at launchpad $l$ \\
    $\omega_{lr}$ & Index of reliability level of launchpad $l$\\
    $w_l^a$ & Average waiting time for a drone at launchpad $l$ \\
    $N_Z^I$ & Number of idle drivers in zone $Z$ \\
    $p_{t_{(i,j,m)}^n}$ & The probability of a taker in state 
    $t_{(i,j,n)}^m$ to be matched with a seeker\\
    $p_{s_{(i,j)}}$ & The probability for a seeker $s_{(i,j)}$ to be matched with a taker \\
    $\rho_{t_{(i,j,n)}^m}$ & The probability of having a taker in state $t_{(i,j,m)}^n$\\
    $\tau_{t_{(i,j,n)}^m}$ & Average time for a taker in state $t_{(i,j,n)}^m$ available for matching \\
    $\lambda_{t_{(i,j,n)}^m}^{um}$ & Arrival rate of unmatched takers in state $t_{(i,j,n)}^m$\\
    $t_{s_{(i,j)}}^{taker}$ & Expected delivery time for a seeker $s_{(i,j)}$ getting matched with a taker\\
    $t_{s_{(i,j)}}^{idle}$ & Expected delivery time for a seeker $s_{(i,j)}$ getting matched with an idle courier\\
    $w_{ij}^g$ & Expected delivery time for ground delivery services\\
    $w_{ilkj}^a$ & Expected delivery time for air delivery services \\
    $t_{ij}^{s1}$/$t_{ij}^{s2}$ & Average time for an order delivered from node $i$ to node $j$ shared with one/two additional order(s) \\
    $q_{Z_1Z_2}^{I,r}$ & Flow of idle drivers repositioning from zone $Z_1$ to zone $Z_2$ \\
    $N^r$ & Number of repositioning couriers \\
    $N^c$ & Number of occupied couriers \\
    $N^a$ & Number of drones \\
    $\Theta_{ij}^1$/$\Theta_{ij}^2$ & Predicted values of the neural network for OD pair $(i,j)$ \\
    $h_{ij}$ & the prediction values of each neuron in the neural network \\
    $\check h_{ij}$, $\kappa_{ij}$ & Linearization parameters of the neural network formulation \\
    \hline\hline
    \caption{Summary of Notations}\\
        \label{tab:notation}\\
	\end{longtable}
\end{center}

\end{appendices}

\bibliographystyle{agsm}
\bibliography{references}

\end{document}